\RequirePackage{fix-cm}
\documentclass[12pt]{article} 
\pdfoutput=1
\usepackage{amsthm}
\usepackage[a4paper,margin=1in]{geometry}
\usepackage{lmodern}
\usepackage{mathptmx} 
\usepackage[english]{babel}
\usepackage[utf8]{inputenc}
\usepackage[T1]{fontenc}
\usepackage{amsfonts}
\usepackage{amsmath}
\usepackage{amssymb}
\usepackage{latexsym}
\usepackage{bm}
\usepackage{capt-of}
\usepackage{graphicx}
\usepackage{xspace}
\usepackage{color} 
\usepackage{url}

\newtheorem{theorem}{Theorem}

\newtheorem{lemma}{Lemma}

\newcommand{\etal}{\emph{et~al.}\xspace}

\newcommand{\flp}{\mbox{\rm FLP}\xspace} 
\newcommand{\mflp}{\mbox{\rm MFLP}\xspace} 
\newcommand{\smflp}{\mbox{\rm SMFLP}\xspace} 
\newcommand{\seflp}{\mbox{\rm E$^2$FLP}\xspace} 

\newcommand{\facloc}{{\sc Facility Location Problem}\xspace} 
\newcommand{\metfacloc}{{\sc Metric FLP}\xspace} 
\newcommand{\squmetfacloc}{{\sc Squared Metric FLP}\xspace} 
\newcommand{\eucfacloc}{{\sc Euclidean FLP}\xspace} 
\newcommand{\squeucfacloc}{{\sc Squared Euclidean FLP}\xspace} 

\newcommand{\Oh}{\textrm{O}}

\newcommand{\PP}{\mbox{\rm P}\xspace}
\newcommand{\NP}{\mbox{\rm NP}\xspace}
\newcommand{\DTIME}{\mbox{\rm DTIME}\xspace}

\newcommand{\ceil}[1]{\lceil{#1}\rceil}

\newcommand{\RNN}{{\mathbb{R}_{+}}} 
\newcommand{\RP}{{\mathbb{R}_{+}^{*}}} 

\newcommand{\coefa}[2]{\textrm{coeff}_#1[\alpha_#2]}
\newcommand{\coefd}[2]{\textrm{coeff}_#1[-d_#2]}
\newcommand{\coeff}[1]{\textrm{coeff}_#1[f]}
\newcommand{\coef}[1]{\textrm{coeff}[#1]}

\newcommand{\optval}[3]{{#1}^{\scriptscriptstyle #2}_{#3}}

\newcommand{\zz}[2]{\optval{z}{#1}{#2}}
\newcommand{\ww}[2]{\optval{w}{#1}{#2}}
\newcommand{\xx}[2]{\optval{x}{#1}{#2}}

\newcommand{\wwr}[2]{\optval{\dot{w}}{}{#2}}
\newcommand{\xxr}[2]{\optval{\dot{x}}{}{#2}}

\newcommand{\zzm}[2]{\optval{\hat{z}}{#1}{#2}}

\newcommand{\xxm}[2]{\optval{\hat{x}}{#1}{#2}}
  
\newcommand{\cali}{\mathcal{I}}
\newcommand{\calj}{\mathcal{J}}
\newcommand{\cals}{\mathcal{S}}
\newcommand{\calu}{\mathcal{U}}

\newcommand{\ox}{\overline{x}}
\newcommand{\oy}{\overline{y}}

\newcommand{\djf}{d_j}       
\newcommand{\djc}{d^{(c)}_j}
\newcommand{\djd}{d^{(d)}_j}
\newcommand{\djm}{d^{(\max)}_j}
\newcommand{\djlc}{d^{(c)}_{j'}}
\newcommand{\djlm}{d^{(\max)}_{j'}}
\newcommand{\djjl}{d_{jj'}}

\newcommand{\mystyle}{\textstyle}
\newcommand{\suml}{\sum\limits}

\definecolor{magiccolor}{rgb}{0.72,0.56,0.06}
\newcommand{\MAGIC}[1]{{#1}} 

\begin{document}
\title{A Systematic Approach to Bound Factor-Revealing LPs and its Application to 
the Metric and Squared Metric Facility Location Problems\footnotemark[1]}
\author{%
C.G. Fernandes\footnotemark[2],
L.A.A. Meira\footnotemark[3],
F.K. Miyazawa\footnotemark[4],
and L.L.C. Pedrosa\footnotemark[4]}
\date{August 6, 2013}

\renewcommand{\thefootnote}{\fnsymbol{footnote}}
\footnotetext[1]{This research was partially supported by CNPq
(grant numbers 306860/2010-4, 473867/2010-9, and 309657/2009-1),
FAPESP (grant number 2010/20710-4) and Project MaCLinC of NUMEC/USP.}
\footnotetext[2]{Institute of Mathematics and Statistics, University of S\~ao Paulo, Brazil, \texttt{cris@ime.usp.br}.}
\footnotetext[3]{Faculty of Technology, University of Campinas, Brazil, \texttt{meira@ft.unicamp.br}}
\footnotetext[4]{Institute of Computing, University of Campinas, Brazil, \texttt{fkm@ic.unicamp.br}, \texttt{lehilton@ic.unicamp.br}.}

\setcounter{page}{0}
\maketitle
\begin{abstract}
A systematic technique to bound factor-revealing linear programs is
presented.  We show how to derive a family of \emph{upper bound}
factor-revealing programs (UPFRP), and show that each such program can
be solved by a computer to bound the approximation factor of an
associated algorithm. Obtaining an UPFRP is straightforward, and can
be used as an alternative to analytical proofs, that are usually very
long and tedious. We apply this technique to the Metric Facility
Location Problem (\mflp) and to a generalization where the distance
function is a squared metric. We call this generalization the Squared
Metric Facility Location Problem (\smflp) and prove that there is no
approximation factor better than $2.04$, assuming~$\PP\neq\NP$.  Then,
we analyze the best known algorithms for the \mflp based on
primal-dual and LP-rounding techniques when they are applied to the
\smflp. We prove very tight bounds for these algorithms, and show that
the LP-rounding algorithm achieves a ratio of $2.04$, and therefore
has the best factor for the \smflp. We use UPFRPs in the dual-fitting
analysis of the primal-dual algorithms for both the \smflp and the
\mflp, improving some of the previous analysis for the \mflp.
\end{abstract}

\thispagestyle{empty}
\newpage

\section{Introduction}

Let $C$ and $F$ be finite disjoint sets. Call \emph{cities} the elements of
$C$ and \emph{facilities} the elements of~$F$. For each facility $i$ and city
$j$, let $c_{ij}$ be a non-negative number representing the cost to connect~$i$
to~$j$. Additionally, let $f_i$ be a non-negative number representing the cost
to open facility~$i$. For each city $j$ and subset $F'$ of~$F$, let $c(F',j) =
\min_{i \in F'} c_{ij}$. The \facloc (\flp) consists of the following: given
sets~$C$ and $F$, and $c$ and $f$ as above, find a subset $F'$ of $F$ such that
$\sum_{i \in F'} f_i + \sum_{j \in C} c(F',j)$ is minimum.
Hochbaum~\cite{Hochbaum82a} presented an $\Oh(\log n)$-approximation for the
\flp.

A well-studied particular case of the \flp is its so called metric variant. We
say that an instance $(C,F,c,f)$ of the \flp is \emph{metric} if $c_{ij} \leq
c_{ij'} + c_{i'j'} + c_{i'j}$, for all facilities $i$ and~$i'$, and cities $j$
and~$j'$. In the context of the \flp, this inequality is called the
\emph{triangle inequality}. The \metfacloc, denoted by \mflp, is the
particular case of the \flp that considers only metric instances. Several
algorithms were proposed in the literature for the
\mflp~\cite{ByrkaA2010,ChudakS99,GuhaK98,JainMMSV03,JainV01,Li2011,MahdianYZ06,
  ShmoysTA97}. In particular, the best known algorithm for the \mflp is a
$1.488$-approximation proposed by Li~\cite{Li2011}. Also, Guha and
Khuller~\cite{GuhaK99} proved an inapproximability result that states that
there is no approximation algorithm for the \mflp with a ratio smaller than
$1.463$, unless $\NP \subseteq \DTIME[n^{\Oh(\log \log n)}]$. This result was
strengthened by Sviridenko, who showed that the lower bound holds unless $\PP
= \NP$ (see~\cite{Vygen2005}).

The \eucfacloc is a particular case of the \mflp also considered in the
literature. In the \eucfacloc, one is given a position in an Euclidean space
for each city and for each facility, and the cost $c_{ij}$ is the Euclidean
distance between the position of facility~$i$ and the position of city $j$.
There is a PTAS for the Euclidean \flp in 2-dimensional space, by Arora,
Raghavan, and Rao~\cite{AroraRR98}.

Yet another variant considered in the literature is the so called \squeucfacloc,
denoted here by \seflp. In this variant, as in the Euclidean case, one is given
a position in an Euclidean space for each city and for each facility. Here, the
cost $c_{ij}$ is the square of the Euclidean distance between the position of
facility~$i$ and the position of city~$j$. This cost measure is known as
$\ell_2^2$, and was, for instance, considered by Jain and
Vazirani~\cite[pp.~292--293]{JainV01} in the context of the \flp. Their approach
implies a 9-approximation for the \seflp.

We consider instances $(C,F,c,f)$ of the \flp such that a relaxed version of the
triangle inequality is satisfied. We say that a cost function $c$ is a
\emph{squared metric}, if, for all facilities~$i$ and~$i'$, and cities~$j$ and
$j'$, we have $\sqrt{c_{ij}}\leq\sqrt{c_{ij'}} + \sqrt{c_{i'j'}} +
\sqrt{c_{i'j}}$. The particular case of the \flp that only considers instances
with a squared metric is called \squmetfacloc, and is denoted by \smflp. Notice
that the \smflp is a generalization of the \seflp and of the \mflp. Thus any
approximation for the \smflp is also an approximation for the \seflp or the
\mflp, and the inapproximability results for the \mflp are also valid for the
\smflp.
The 9-approximation of Jain and Vazirani~\cite{JainV01} applies also to the
\smflp and, to our knowledge, it has the best previously known approximation
factor. The choice of squared metrics discourages excessive distances in the
solution. This effect is important in several applications, such as $k$-means
and classification problems. 

Although there are several algorithms for the \mflp in the literature, there
are very few works on the \smflp. Nevertheless, one may try to solve an
instance of the \smflp using good algorithms designed for the \mflp. Since
these algorithms and their analysis are based on the assumption of the
triangle inequality, it is reasonable to expect that they generate good
solutions also for the \smflp.  However, there is no trivial way to derive an
approximation factor from the \mflp to the \smflp, so each algorithm must be
reanalyzed individually. In this paper, we analyze three primal-dual
algorithms (the $1.861$ and the $1.61$-approximation algorithms of
Jain~\etal~\cite{JainMMSV03}, and the $1.52$-approximation of Mahdian, Ye, and
Zhang~\cite{MahdianYZ06}) and an LP-rounding algorithm (Chudak and Shmoys's
algorithm~\cite{Chudak2004} used in the $1.5$-approximation of Byrka and
Aardal~\cite{ByrkaA2010}) when applied to \smflp instances.  We show that
these algorithms achieve ratios of \MAGIC{$2.87$}, \MAGIC{$2.43$},
\MAGIC{$2.17$}, and \MAGIC{$2.04$} for the \smflp, respectively.  
The last approximation factor is the best possible, as we show a
$2.04$-inapproximability limit for the \smflp. This was obtained by extending
the metric case hardness results of Guha and Khuller~\cite{GuhaK99}.

The original analysis of the three primal-dual algorithms are based on the so
called families of \emph{factor-revealing linear
  programs}~\cite{JainMMSV03,MahdianYZ06}.  The value of a computer calculated
optimal solution for any program in this family gives a lower bound on the
approximation factor.  An upper bound, however, is obtained analytically by
bounding the value of every program in this family, which requires long and
tedious proofs.  In this paper, we propose a way to obtain a new family of
upper bound factor-revealing programs, as an alternative technique to achieve
an upper bound.  Now, the upper bound on the approximation factor is also
obtained by a computer calculated solution of a single program. We note that,
for the \smflp, our factor-revealing programs are nonlinear, since the squared
metric constraints contain square roots. We tackle this by replacing these
constraints with an infinite set of linear constraints.

Recently, Mahdian and Yan~\cite{MahdianY2011} introduced the \emph{strongly
factor-revealing linear programs}. Our upper bound factor-revealing program is
similar to a strongly factor-revealing program. The techniques involved in
obtaining our program, however, are different. To obtain a strongly
factor-revealing linear program, one projects a solution of an arbitrarily large
linear program into a linear program with a constant number of variables, and
guesses how to adjust the restrictions to obtain a feasible solution. In our
approach, we define a candidate dual solution for a program with a fixed number
of variables, and obtain an upper bound factor-revealing program directly in the
form of a minimization program using only straightforward calculations. For the
case of the \smflp, we observed that calculating the dual upper bound program is
easier than projecting the solutions on the primal. Also, we have considered the
case of the \mflp, for which the obtained lower and upper bound factor-revealing
programs converge.

Our contribution is two-folded. First, we make an important step towards
generalizing the squared Euclidean distance and successfully analyze this
generalization in the context of the \flp. Second, more importantly, we propose a
new technique to systematically bound factor-revealing programs. This technique
is used in the dual-fitting analysis of the primal-dual algorithms for both the
\smflp and the \mflp. We hope that this technique can also be used in the
analysis of other dual-fitting algorithms analyzed through factor-revealing LPs.

The paper is organized as follows. In Section~\ref{sec:newanalysis},
we present the new technique analyzing the performance of the first
algorithm of Jain \etal~\cite{JainMMSV03} for the \smflp.
Section~\ref{sec:segundo} applies the technique in the analysis of the
second algorithm of algorithm of Jain \etal~\cite{JainMMSV03}.
Section~\ref{sec:MahdianYZ} analyzes the performance of the algorithm
of Mahdian, Ye, and Zhang~\cite{MahdianYZ06} for the \smflp. In
Section~\ref{sec-optimal}, we present the analysis of the algorithm of 
Li~\cite{Li2011}, which happens to be optimal for \smflp, according 
to the complexity result we present in Section~\ref{sec:complexity}.
Finally we make some concluding remarks in Section~\ref{sec:conclu}.

\section{A new factor-revealing analysis}\label{sec:newanalysis}

We analyze the algorithms of Jain \etal~\cite{JainMMSV03} using
a new systematic factor-revealing technique.
For each algorithm, Jain \etal~\cite{JainMMSV03} analysis uses a
family of factor-revealing LPs parameterized by some $k$. The optimal
value $z_k$ of the corresponding LP in the family is such that
$\sup_{k\ge1} z_k$ is the approximation factor of the algorithm. 
Thus each value $z_k$ is a \emph{lower} bound on the approximation
factor and one has to analytically upper bound $\sup_{k\ge1} z_k$ to
obtain an approximation factor. This is a nontrivial analysis, since
it is done by guessing a general suboptimal dual solution for the LP,
usually inspired by numerically obtained dual LP solutions for small
values of $k$.
In this section, we show how to derive a family of \emph{upper
  bound factor-revealing programs} (UPFRP) parameterized by some $t$,
so that, for any given~$t$, the optimal value $x_t$ of one such
program is an upper bound on $\sup_{k\ge1} z_k$. Obtaining a UPFRP
and solving it using a computer is much simpler and more
straightforward than using an analytical proof to obtain the
approximation factor, since this does not include a guessing step and
a manual verification of the feasibility of the
solution. Additionally, as a property of the UPFRPs, we may tighten
the obtained factor by solving the LP for larger values of~$t$. In
fact, in some cases (see Theorem~\ref{theorem-alg1-conv} below), the
lower and upper bound factor-revealing programs converge, that is,
$\sup_{k\ge1} z_k = \inf_{t \ge 1} x_t$.
%

We use a UPFRP to show that, when applied to \smflp instances, the
first algorithm of Jain \etal~\cite{JainMMSV03}, denoted by $A1$,
is a \MAGIC{$2.87$}-approximation.
For the sake of completeness, the algorithm is described in the following.

\paragraph{Algorithm $A1$ $(C,F,c,f)$~\cite{JainMMSV03}}
\begin{enumerate}
\item Set $U:=C$, meaning that every facility starts unopened, and every city
unconnected. Each city $j$ has some budget $\alpha_j$, initially 0, and, at
every moment, the budget that an unconnected city $j$ offers to some unopened
facility $i$ equals to $\max(\alpha_j - c_{ij}, 0)$.
\item While $U \ne \emptyset$, the budget of each unconnected city is increased
continuously until one of the following events occur:
  \begin{enumerate}
    \item For some unconnected city $j$ and some open facility $i$, $\alpha_j = c_{ij}$.
    In this case, connect city $j$ to facility $i$ and remove $j$ from $U$.
    \item For some unopened facility $i$, $\sum_{j \in U} \max(\alpha_j - c_{ij},0) = f_i$.
    In this case, open facility~$i$ and, for every unconnected city $j$ with
    $\alpha_j \ge c_{ij}$, connect $j$ to $i$ and remove $j$ from~$U$.
  \end{enumerate}
\end{enumerate}

\medskip

The analysis presented by Jain \etal~\cite{JainMMSV03} uses the dual fitting
method. That is, their algorithms produce not only a solution for the \mflp, but
also a vector $\alpha=(\alpha_1,\ldots,\alpha_{|C|})$ such that the value of the
solution produced is equal to $\sum_j \alpha_j$. Moreover, for the first
algorithm, following the dual fitting method, Jain \etal~\cite{JainMMSV03}
proved that the vector $\alpha/1.861$ is a feasible solution for the dual linear
program presented as~(3) in~\cite{JainMMSV03}, concluding that the algorithm is
a $1.861$-approximation for the \mflp. To present a similar analysis for the
\smflp, we use the same definitions and follow the steps of Jain \etal~analysis.
We start by adapting Lemma~3.2 from~\cite{JainMMSV03} for a squared metric.

\newcommand{\PasteLemmaAlgOneMetric}{
For every facility $i$, cities $j$ and $j'$, and vector $\alpha$ obtained by the
first algorithm of Jain~\etal~\cite{JainMMSV03} given an instance of the \smflp,
$
\sqrt{\alpha_j} \le \sqrt{\alpha_{j'}} +\sqrt{c_{ij'}} + \sqrt{c_{ij}}.
$%
}
\begin{lemma}\label{lemma-alg1-metric}
\PasteLemmaAlgOneMetric
\end{lemma}
\begin{proof}
If $\alpha_j \leq \alpha_{j'}$, the inequality obviously holds.  So assume
$\alpha_j > \alpha_{j'}$. Let $i'$ be the facility to which the algorithm
connects city $j'$. Thus $\alpha_{j'} \ge c_{i'j'}$ and facility $i'$
is open at time $\alpha_{j'} <  \alpha_j$. If $\alpha_j > c_{i'j}$, then city
$j$ would have connected to facility $i'$ at some time $t \le
\max(\alpha_{j'},c_{i'j}) < \alpha_j$, and $\alpha_j$ would have stopped
growing then, a contradiction. Hence $\alpha_j \leq c_{i'j}$.
Furthermore, by the squared metric constraint,
$\sqrt{c_{i'j}} \leq \sqrt{c_{i'j'}} + \sqrt{c_{ij'}} + \sqrt{c_{ij}}$.
Therefore $\sqrt{\alpha_j} \le \sqrt{\alpha_{j'}} +\sqrt{c_{ij'}} + \sqrt{c_{ij}}$.
\end{proof}

A facility $i$ is said to be $\gamma$-\emph{overtight} for some
positive~$\gamma$ if, at the end of the algorithm,
\begin{equation} \label{inequality-gamma}
 \sum_j \max\big(\frac{\alpha_j}{\gamma} - c_{ij},0\big) \ \leq \ f_i.
\end{equation}
Observe that, if every facility is $\gamma$-overtight, then the vector
$\alpha/\gamma$ is a feasible solution for the dual linear program presented
as~(3) in~\cite{JainMMSV03}. Jain~\etal\ proved that, for the \mflp, every
facility is $1.861$-overtight. We want to find a $\gamma$ for the \smflp, as
close to $1$ as possible, for which every facility is $\gamma$-overtight.

Fix a facility $i$. Let us assume without loss of generality that
$\alpha_j\geq \gamma\,c_{ij}$ only for the first~$k$ cities.  Following the
lines of Jain \etal~\cite{JainMMSV03}, we want to obtain the so called
\emph{(lower bound) factor-revealing} program. We define a set of
variables $f$, $d_j$, and $\alpha_j$, corresponding to facility cost $f_i$,
distance~$c_{ij}$, and city contribution $\alpha_j$.  Then, we capture the
intrinsic properties of the algorithm using constraints over these
variables. We assume without loss of generality that $\alpha_1 \le \cdots \le
\alpha_k$.  Also, we use Lemma~3.3 from~\cite{JainMMSV03}, that states that
the total contribution offered to a facility at any time is at most its cost,
that is, $\sum_{l = j}^k \max(\alpha_j - d_l, 0) \le f$. Additionally, we have
the inequalities from Lemma~\ref{lemma-alg1-metric}. Subject to all of these
constraints, we want to find the minimum~$\gamma$ such that the facility is
$\gamma$-overtight. In terms of the defined variables, we want the maximum
ratio $\sum_{j=1}^k \alpha_j/(f+\sum_{j=1}^k d_j)$. We obtain the following
lower bound factor-revealing program:
\begin{equation}\label{program-alg1-frlp}
\begin{array}{rlll}
\zz{A1}{k} = & \mbox{\rm max  }  & \frac{\sum_{j=1}^{k} \alpha_j}{f + \sum_{j=1}^{k} d_j}  \\
      & \mbox{\rm s.t.}   & \alpha_j \le \alpha_{j+1}
                          & \quad \forall \; 1 \le j < k \\
      &                   & \sqrt{\alpha_j} \le \sqrt{\alpha_l} + \sqrt{d_j} + \sqrt{d_l}
                          & \quad \forall \; 1 \le j, \, l \le k \\
      &                   & \sum_{l = j}^{k} \max(\alpha_j - d_l, 0) \le f
                          & \quad \forall \; 1 \le j \le k \\
      &                   & \alpha_j, d_j, f \ge 0
                          & \quad \forall \; 1 \le j \le k.
\end{array}
\end{equation}

The next lemma has the same statement of Lemma 3.4 in~\cite{JainMMSV03}, but it
refers to program~\eqref{program-alg1-frlp}. Since the proof is the same, we
omit it.

\begin{lemma}\label{l34}
  Let $\gamma=\sup_{k \geq 1}\zz{A1}{k}$. Every facility is $\gamma$-overtight.
\end{lemma}

Therefore $\sup_{k \geq 1}\zz{A1}{k}$ is an upper bound on the approximation
factor of the algorithm for the \smflp. A slight modification of the example
presented in Theorem 3.5 of~\cite{JainMMSV03} shows that this upper bound is
tight. (Take $c_{ij} = (\sqrt{d_i}+\sqrt{d_j}+\sqrt{\alpha_i})^2$ if $k \geq i 
\neq j$.)

Although the constraints comming from Lemma~\ref{lemma-alg1-metric}
are defined by square roots, they are convex. Indeed, the next lemma,
whose proof is presented in Appendix~\ref{apendice-provas}, 
shows that they can be expressed by an infinite set of linear
inequalities. 


\begin{lemma}\label{corollary-alg1-inequality}
Given an instance of the \smflp, for every facility $i$, cities $j$
and $j'$, the vector $\alpha$ produced by the first algorithm of Jain
\etal~\cite{JainMMSV03} is such that, for every positive $\beta$,
$\gamma$, and~$\delta$, 
$$
\alpha_j \le (1 + \beta  + \frac{1}{\gamma}) \alpha_{j'}
           + (1 + \gamma + \frac{1}{\delta}) c_{ij'}
           + (1 + \delta + \frac{1}{\beta} ) c_{ij}.
$$
\end{lemma}

As a consequence, for any candidate solution for
program~\eqref{program-alg1-frlp} not satisfying a constraint comming
from Lemma~\ref{lemma-alg1-metric}, there exists some violated linear
inequality.

\subsection{A first analysis}

Our first step is to relax~\eqref{program-alg1-frlp} into a linear
program. For that, we adjust the objective function as
in~\cite{JainMMSV03}, and we approximate the inequalities with square
roots using inequalities given by Lemma~\ref{corollary-alg1-inequality}. 
For simplicity, here we will use only the inequalities corresponding
to $\beta=\gamma=\delta=1$. With this, we will prove that $\sup_{k
  \geq 1}\zz{A1}{k}$ is not greater than \MAGIC{$3.236$}.  Later, we
will improve the obtained result by using a whole set of inequalities
from Lemma~\ref{corollary-alg1-inequality}, and using a more standard
factor-revealing analysis for the \smflp. The relaxed lower
 factor-revealing linear program is:
\begin{equation}\label{program-alg1-reveal-tres}
\begin{array}{rlll}
\wwr{A1}{k} = & \mbox{\rm max  }  & \sum_{j=1}^{k} \alpha_j \\
      & \mbox{\rm s.t.}  & f + \sum_{j=1}^k d_j \le 1 & \\
      &                        & \alpha_j \le \alpha_{j+1}
                               & \quad \forall \; 1 \le j < k \\
      &                        & \alpha_j \le 3 \alpha_l + 3 d_j + 3 d_l
                               & \quad \forall \; 1 \le j, \, l \le k \\
      &                        & x_{jl} \ge  \alpha_j - d_l
                               & \quad \forall \; 1 \le j \le l \le k \\
      &                        & \sum_{l = j}^{k} x_{jl} \le f
                               & \quad \forall \; 1 \le j \le k \\
      &                        & \alpha_j, d_j, f, x_{jl} \ge 0
                               & \quad \forall \; 1 \le j \le l \le k.
\end{array}
\end{equation}

As \eqref{program-alg1-reveal-tres} is a relaxation of
\eqref{program-alg1-frlp}, we have that $\zz{A1}{k} \leq \wwr{A1}{k}$ and thus
an upper bound on $\sup_{k \geq 1}\wwr{A1}{k}$ is also an upper bound on
$\sup_{k \geq 1}\zz{A1}{k}$. Solving linear
program~\eqref{program-alg1-reveal-tres} using CPLEX for \MAGIC{$k=540$}, we
obtain the next lemma.

\begin{lemma}\label{lemma-alg1-lower-tres}
$\sup_{k \ge 1} \wwr{A1}{k} \ge \MAGIC{3.220}$.
\end{lemma}

To obtain an upper bound on their factor-revealing linear program,
Jain~\etal~\cite{JainMMSV03} presented a general dual solution of a
relaxed version of the lower bound factor-revealing linear
program. This solution is deduced from computational experiments and
empirical results for small values of~$k$. In their analysis, they
guessed 2- and 3-step functions for a set of dual variables, and used
a long verification to show that the value of such solution was not
greater than $1.861$. For the squared metric case, if we use step
functions for the dual variables, the bound on the factor would be as
bad as $3.625$. One can improve the obtained factor to \MAGIC{$3.512$}
by guessing a piecewise function whose pieces are either constants or
hyperboles.

Instead of looking for a good general dual solution, we use an alternative 
analysis and 
derive a linear minimization program from~\eqref{program-alg1-reveal-tres}
whose feasible solutions are upper bounds on $\sup_{k \ge 1} \wwr{A1}{k}$. 
Afterwards, we give an upper bound on the approximation factor by 
presenting a feasible solution for this program of value less than 
\MAGIC{$3.236$}.

The idea is to determine a conical combination of the inequalities
of~\eqref{program-alg1-reveal-tres} that imply
inequality~\eqref{inequality-gamma} for a $\gamma$ as small as possible.  The
linear minimization program will help us to choose the coefficients of such
conical combination.

First, rewrite the third inequality of program~\eqref{program-alg1-reveal-tres},
so that the right-hand side is zero. For each $j$ and $l$, we multiply
the corresponding inequality by $\varphi_{jl}$. Denote by $A$ the sum of all
these inequalities, that is,
\begin{equation*}\label{eq-alg1-323-a}
 \sum_{j=1}^{k} \sum_{l = 1}^{k} \varphi_{jl} (\alpha_j - 3 \alpha_l - 3 d_l - 3 d_j) \le 0.
\end{equation*}


The fourth and fifth inequalities of program~\eqref{program-alg1-reveal-tres}
can be relaxed to the set of inequalities $\sum_{i = j}^{l} (\alpha_j - d_i)
\le f$, one for each $l$ such that $j \le l \le k$. For each~$j$ and~$l$,
we multiply the corresponding inequality by $\theta_{jl}$ and denote by $B$
the inequality resulting of summing them up, that is,
\begin{equation*}\label{eq-alg1-323-b}
 \sum_{j=1}^{k} \sum_{l = j}^{k} \theta_{jl} \sum\nolimits_{i = j}^{l} (\alpha_j - d_i)
 \le \left(  \sum_{j=1}^{k} \sum_{l = j}^{k} \theta_{jl} \right) f.
\end{equation*}

The coefficients of $\alpha_j$ in $A$ and $B$ are, respectively,
\begin{equation*}\label{eq-alg1-323-coefa}
  \coefa{A}{j} = \sum_{l=1}^{k} (\varphi_{jl} - 3 \varphi_{lj})
  \qquad\mbox{ and }\qquad
  \coefa{B}{j} = \sum_{l=j}^k (l - j + 1) \theta_{jl},
\end{equation*}
and the coefficients of $-d_j$ in $A$ and $B$ are, respectively,
\begin{equation*}\label{eq-alg1-323-coefd}
  \coefd{A}{j} = \sum_{l=1}^{k} 3 (\varphi_{jl} + \varphi_{lj})
  \qquad\mbox{ and }\qquad
  \coefd{B}{j} = \sum_{i=1}^{j} \sum_{l = j}^{k} \theta_{il}.
\end{equation*}

Now, we sum inequalities $A$ and $B$ and obtain a new inequality $C$:
\begin{equation}\label{eq-alg1-323-c}
 \sum_{j=1}^{k} \coefa{C}{j} \; \alpha_j - \sum_{j=1}^{k} \coefd{C}{j} \; d_j \le \coeff{C} \; f.
\end{equation}

We want to find values for $\gamma$, $\theta_{jl}$, and $\varphi_{jl}$ so that
the corresponding coefficients of~$C$ are such that inequality~\eqref{eq-alg1-323-c}
implies, for sufficiently large $k$, that
\begin{equation}\label{eq-alg1-323-objective}
 \sum_{j=1}^{k} \alpha_j - \gamma \sum_{j=1}^{k} d_j \le \gamma f.
\end{equation}
Moreover, we want $\gamma$ as small as possible.
To obtain inequality~\eqref{eq-alg1-323-objective} from
inequality~\eqref{eq-alg1-323-c}, it is enough that, for each $j$, coefficient
$\coefa{C}{j} \ge 1$, $\coefd{C}{j} \le \gamma$, and $\coeff{C} \le \gamma$.
Hence, this can be expressed by the following linear program.
\begin{equation}\label{program-alg1-dual-tres}
\begin{array}{rlll}
y_k \ = & \mbox{\rm min }   & \gamma  \\
      & \mbox{\rm s.t.}  & \coefa{C}{j} \ge 1
                               & \quad \forall \; 1 \le j \le k \\
      &                        & \coefd{C}{j} \le \gamma
                               & \quad \forall \; 1 \le j \le k \\
      &                        & \coeff{C} \le \gamma
                               & \\
      &                        & \varphi_{jl} \ge 0
                               & \quad \forall \; 1 \le j, \, l \le k \\
      &                        & \theta_{jl} \ge 0
                               & \quad \forall \; 1 \le j \le l \le k.
\end{array}
\end{equation}

The interested reader may observe that program~\eqref{program-alg1-dual-tres}
is the dual of a relaxed version of the lower bound factor-revealing
linear program~\eqref{program-alg1-reveal-tres}.  Therefore, its optimal value
is an upper bound on the optimal value of~\eqref{program-alg1-reveal-tres}, 
that is, $\wwr{A1}{k} \le y_k$ for every $k$.



\newcommand{\PasteLemmaAlgOneThreeTwoThree}{
$\sup_{k \geq 1}\wwr{A1}{k} \le \MAGIC{3.236}$.
}
\begin{lemma}\label{lemma-alg1-323}
\PasteLemmaAlgOneThreeTwoThree
\end{lemma}
\begin{proof}
  We start by observing that $\sup_{k \geq 1}\wwr{A1}{k}$ does not decrease if
  we restrict attention to values of $k$ that are multiples of a fixed
  positive integer~$t$. Indeed, for an arbitrary positive integer~$p$, by
  making $t$ replicas of a solution of~\eqref{program-alg1-reveal-tres}
  for~$k=p$, and scaling the variables by~$1/t$, we obtain a solution
  of~\eqref{program-alg1-reveal-tres} for~$k=pt$, that is, we deduce that
  $\wwr{A1}{p} \le \wwr{A1}{pt}$. So we may assume that $k$ has the form $k =
  pt$ with $p$ and $t$ positive integers, and our goal is to prove that 
  $\wwr{A1}{k} \le \MAGIC{3.236}$.

  We will use program~\eqref{program-alg1-dual-tres} to obtain a tight
  upper bound on $\wwr{A1}{k}$. The size of this program however
  depends on $k$, which can be arbitrarily large. So we will use a
  scaling argument to create another linear minimization program with
  a fixed number (depending only on $t$) of variables, and obtain a
  feasible solution for program~\eqref{program-alg1-dual-tres} from a
  solution for this smaller program. Then, we will show that the value
  of the generated solution for~\eqref{program-alg1-dual-tres} is
  bounded by the value of the small solution.

  Consider variables $\gamma' \in \RNN$, $\varphi'_{jl} \in \RNN$ for $1 \le
  j, \, l \le t$, and $\theta'_{jl} \in \RNN$ for $1 \le j \le l \le t$. For
  an arbitrary $n$, let $\hat{n} = \lceil \frac{n}{p} \rceil$. We obtain a
  candidate solution for program~\eqref{program-alg1-dual-tres} by taking
\begin{equation}\label{eq-alg1-323-vardef}
  \varphi_{jl} = \frac{\varphi'_{\hat{j} \hat{l}}}{p}
  \mbox{,}\qquad
  \theta_{jl} = \frac{\theta'_{\hat{j}\hat{l}}}{p^2}\mbox{,}
  \qquad\mbox{and}\qquad
  \gamma = \gamma'.
\end{equation}

Let us calculate each coefficient of $C$ for this solution.

\begin{align*}
 \coefa{C}{j} &=   \sum_{l=1}^{k} (\varphi_{jl} - 3 \varphi_{lj})
                 + \sum_{l=j}^{k} (l - j + 1) \theta_{jl} \\
              &=   \sum_{l=1}^{k} (\frac{\varphi'_{\hat{j}\hat{l}}}{p} - 3 \frac{\varphi'_{\hat{l}\hat{j}}}{p})
                 + \sum_{l=j}^{k} (l - j + 1) \frac{\theta'_{\hat{j}\hat{l}}}{p^2} \\
              &\ge \sum_{l=1}^{pt} (\frac{\varphi'_{\hat{j}\hat{l}}}{p} - 3 \frac{\varphi'_{\hat{l}\hat{j}}}{p})
                 + \sum_{l=p\,\hat{j} + 1}^{pt} (l - p\,\hat{j}) \frac{\theta'_{\hat{j}\hat{l}}}{p^2} \\
              &=   \sum_{l'=1}^{t} p (\frac{\varphi'_{\hat{j}l'}}{p} - 3 \frac{\varphi'_{l'\hat{j}}}{p})
                 + \sum_{l'=\hat{j} + 1}^{t} \frac{\theta'_{\hat{j}l'}}{p^2} \sum_{i = 0}^{p - 1} (p\,l' - i - p\,\hat{j}) \\
              &=   \sum_{l'=1}^{t} (\varphi'_{\hat{j}l'} - 3 \varphi'_{l'\hat{j}})
                 + \sum_{l'=\hat{j} + 1}^{t} \frac{\theta'_{\hat{j}l'}}{p^2} (p^2\,l' - \frac{p(p-1)}{2} - p^2\,\hat{j}) \\
              &\ge \sum_{l'=1}^{t} (\varphi'_{\hat{j}l'} - 3 \varphi'_{l'\hat{j}})
                 + \sum_{l'=\hat{j} + 1}^{t} (l' - \hat{j} - \frac{1}{2}) \theta'_{\hat{j}l'}.
\end{align*}

\begin{align*}
 \coefd{C}{j} &= \sum_{l=1}^{k} 3 (\varphi_{jl} + \varphi_{lj})
               + \sum_{i=1}^{j} \sum_{l = j}^{k} \theta_{il} \\
              &= \sum_{l=1}^{pt} 3 (\frac{\varphi'_{\hat{j} \hat{l}}}{p} + \frac{\varphi'_{\hat{l} \hat{j}}}{p})
               + \sum_{i=1}^{j} \sum_{l = j}^{pt} \frac{\theta'_{\hat{i}\hat{l}}}{p^2} \\
              &\le \sum_{l'=1}^{t} p \cdot 3 (\frac{\varphi'_{\hat{j} l'}}{p} + \frac{\varphi'_{l' \hat{j}}}{p})
               + \sum_{i'=1}^{\hat{j}} p \cdot \sum_{l' = \hat{j}}^{t} p \cdot \frac{\theta'_{i'l'}}{p^2} \\
              &= \sum_{l'=1}^{t} 3 (\varphi'_{\hat{j} l'} + \varphi'_{l' \hat{j}})
               + \sum_{i'=1}^{\hat{j}} \sum_{l' = \hat{j}}^{t} \theta'_{i'l'}.
\end{align*}

\begin{align*}
 \coeff{C}    &=    \sum_{j=1}^{k} \sum_{l= j}^{k} \theta_{jl}
               =    \sum_{j=1}^{pt} \sum_{l = j}^{pt} \frac{\theta'_{\hat{j}\hat{l}}}{p^2}
               \le  \sum_{j'=1}^{t} p \cdot \sum_{l'=\hat{j}}^{t} p \cdot \frac{\theta'_{j'l'}}{p^2}
               =    \sum_{j'=1}^{t} \sum_{l' = \hat{j}}^{t} \theta'_{j'l'}.
\end{align*}

Now, we want to find the minimum value of $\gamma'$ and values for $\varphi'_{jl}$ and
$\theta'_{jl}$ such that the candidate solution for program~\eqref{program-alg1-dual-tres}
is feasible. We may define the following linear program,
named the \emph{upper bound factor-revealing program}.
\begin{equation}\label{program-alg1-upfrlp-tres}
\begin{array}{rllll}
\xxr{A1}{t} = & \mbox{\rm min }   & \gamma'  \\
      & \mbox{\rm s.t.}  & \sum_{l=1}^{t} (\varphi'_{jl} - 3 \varphi'_{lj})
                               + \sum_{l=j + 1}^{t} (l - j - \frac{1}{2}) \theta_{jl} \ge 1
                               & \quad \forall \; 1 \le j \le t \\
      &                        & \sum_{l=1}^{t} 3 (\varphi'_{j l} + \varphi'_{l j})
                               + \sum_{i=1}^{j} \sum_{l = j}^{t} \theta'_{il}   \le \gamma'
                               & \quad \forall \; 1 \le j \le t \\
      &                        & \sum_{j=1}^{t} \sum_{l = j}^{t} \theta'_{jl} \le \gamma'
                               & \\
      &                        & \varphi'_{jl} \ge 0
                               & \quad \forall \; 1 \le j, \, l \le t \\
      &                        & \theta'_{jl} \ge 0
                               & \quad \forall \; 1 \le j  \le l \le t.
\end{array}
\end{equation}

Consider an optimal solution for program~\eqref{program-alg1-dual-tres}.
Replacing it in~\eqref{eq-alg1-323-c}, that is, in inequality $C$, we obtain
$\sum_{j=1}^{k} \alpha_j - \gamma \sum_{j=1}^{k} d_j \le \gamma f$. Thus,
$\wwr{A1}{k} \le \gamma = y_k$. Now, consider an optimal solution for
program~\eqref{program-alg1-upfrlp-tres} and the corresponding generated
solution for program~\eqref{program-alg1-dual-tres}. We obtain $y_k \le \gamma =
\gamma' = \xxr{A1}{t}$ and conclude that $\wwr{A1}{k} \le \xxr{A1}{t}$, 
and that holds for every positive integer $k$.

Using CPLEX to solve program \eqref{program-alg1-upfrlp-tres}, we obtained
\MAGIC{$\xxr{A1}{800} \approx 3.23586 < 3.236$}, and this concludes the proof
of Lemma~\ref{lemma-alg1-323}.
\end{proof}

\subsection{An improved factor-revealing analysis}

In Lemma~\ref{lemma-alg1-323}, we obtained the minimization program
(\ref{program-alg1-upfrlp-tres}) from a conical combination of constraints
from program~\eqref{program-alg1-reveal-tres}
that bounds the approximation factor.
This process is similar to obtaining the dual and using a
scaling argument. Indeed, we propose a systematic way to obtain
an upper bound factor-revealing program.

Consider the dual program of a traditional maximization
factor-revealing linear program for some~$k$. Take $k$ in the form $k
= pt$, for a fixed $t$. We want to create a minimization program that
mimics the dual, but depends only on $t$ and bounds the dual optimal
value for every~$k$. The idea is to constrain the variables of the
small program to obtain a feasible solution for the dual program. To
obtain a linear program independent of $k$, we scale the variables 
by~$p$.  The strategy to obtain an upper bound factor-revealing program
may be summarized as follows:
\begin{enumerate}
  \item obtain the dual $P(k)$ of the lower bound factor-revealing linear program;
  \item consider a block variable $x_i'$ for variables $x_{(i-1)p+1}, \ldots, x_{(i-1)p + p}$ of $P(k)$;
  \item identify each variable $x_i$ with the block variable $x_{\lceil i / p \rceil}'$ scaled by $p$; 
  \item replace identified variables in $P(k)$, canceling factors $p$.
\end{enumerate}
Denote the resulting program by $P'(t)$. If $P'(t)$ depends only on $t$, both
in number of variables and constraints, then any feasible
solution of $P'(t)$ is an upper bound on the solution of $P(pt)$ for every~$p$.
Also, if it is the case that the value of $P(k)$ is not greater than the value of $P(kt)$, for
every $t$, then a solution of $P'(t)$ for any $t$ is also a bound on the
approximation factor. Therefore, we call $P'(t)$ an \emph{upper bound
factor-revealing program}.

\medskip

Although program~\eqref{program-alg1-frlp} is nonlinear, we can still
use the presented strategy.  If the nonlinear constraint is convex, we
can approximate it by using a set of linear inequalities, and
calculate the dual normally. In order to derive a better upper bound
factor-revealing linear program, this time we will use a whole set of
linear inequalities. Consider $m$ tuples $(\beta_i, \gamma_i,
\delta_i)$ of positive real numbers and $B_i = 1 + \beta_i +
\frac{1}{\gamma_i}$, $C_i = 1 + \gamma_i + \frac{1}{\delta_i}$, $D_i =
1 + \delta_i + \frac{1}{\beta_i}$ for $1 \le i \le m$.
Using Lemma~\ref{corollary-alg1-inequality}, we insert inequalities
corresponding to the given tuples, replacing the nonlinear constraint, and
obtain that $\zz{A1}{k} \le \ww{A1}{k}$, where $\ww{A1}{k}$ is given by
\begin{equation}\label{program-alg1-frlp-relax-ineq}
\begin{array}{rlll}
\ww{A1}{k} = & \mbox{\rm max  }  & \sum_{j=1}^{k} \alpha_j \\
      & \mbox{\rm s.t.}  & f + \sum_{j=1}^k d_j \le 1 & \\
      &                        & \alpha_j \le \alpha_{j+1}
                               & \quad \forall \; 1 \le j < k \\
      &                        & \alpha_j \le B_i  \alpha_l + C_i  d_j + D_i  d_l
                               & \quad \forall \; 1 \le j, \, l \le k ,\quad 1 \le i \le m\\
      &                        & x_{jl} \ge  \alpha_j - d_l
                               & \quad \forall \; 1 \le j \le l \le k \\
      &                        & \sum_{l = j}^{k} x_{jl} \le f
                               & \quad \forall \; 1 \le j \le k \\
      &                        & \alpha_j, d_j, f, x_{jl} \ge 0
                               & \quad \forall \; 1 \le j \le l \le k.
\end{array}
\end{equation}

The following lemma gives a lower bound on the approximation factor
of the algorithm for the \smflp using a cutting plane insertion strategy.

\newcommand{\PasteAlgOneFrlpCut}{
$\sup_{k \ge 1} \zz{A1}{k}  \ge \MAGIC{2.86}$.
}
\begin{lemma}\label{lemma-alg1-frlp-cut}
\PasteAlgOneFrlpCut
\end{lemma}
\begin{proof}
Although program~\eqref{program-alg1-frlp} contains nonlinear
constraints, we may use linear program packages to solve it. We start
by solving program~\eqref{program-alg1-frlp-relax-ineq} with a fixed
number of inequalities. Then, we employ a cutting plane insertion
strategy: if the obtained solution violates some inequality with
square roots of~\eqref{program-alg1-frlp}, we derive a cutting plane
using Lemma~\ref{corollary-alg1-inequality}, and resolve the linear
program with this additional constraint. Using CPLEX with the cutting
plane strategy, we obtained \MAGIC{$\zz{A1}{700} \approx 2.86099 >
  2.86$}.  
\end{proof}

Now, we can bound  the approximation factor of the algorithm
using an upper bound factor-revealing program.

\newcommand{\PasteLemmaAlgOneUpfrlpCut}{
$\sup_{k \ge 1} \zz{A1}{k} \le \MAGIC{2.87}$.
}
\begin{lemma}\label{lemma-alg1-upfrlp-cut}
\PasteLemmaAlgOneUpfrlpCut
\end{lemma}
\begin{proof}
It is easy to see that, for program \eqref{program-alg1-frlp-relax-ineq}, 
as in the proof of Lemma~\ref{lemma-alg1-323}, we can restrict
attention to values of $k$ that are multiples of a fixed positive
integer~$t$, that is, $\zz{A1}{k} \le \ww{A1}{kt}$, for every positive 
integer $t$. So we assume that $k$ has the form $k = pt$, with $p$ and 
$t$ positive integers. The dual of the linear
program~\eqref{program-alg1-frlp-relax-ineq} is
\begin{equation}\label{program-alg1-frlp-relax-ineq-dual}\mystyle
\begin{array}{rllll}
\!\!\! \ww{A1}{k} = \mbox{\rm min}    & \gamma  \\
             \mbox{\rm s.t.}   & a_j  - a_{j-1}
                               + \suml_{i=1}^{m}     \suml_{l=1}^{k}                    c_{jli}
                               - \suml_{i=1}^{m} B_i \suml_{l=1}^{k}                    c_{lji}
                               +                     \suml_{l=j}^{k}                    e_{jl}  \ge 1
                               & \ \forall & 1 \le j \le k \\
                               & \suml_{i=1}^{m} C_i \suml_{l=1}^{k} c_{jli}
                               + \suml_{i=1}^{m} D_i \suml_{l=1}^{k} c_{lji}
                               +                     \suml_{l=1}^{j} e_{lj} \le \gamma
                               & \ \forall & 1 \le j \le k \\
                               & \suml_{j=1}^{k} h_j \le \gamma
                               &  \\
                               & e_{jl} \le h_j
                               & \ \forall & 1 \le j \le l \le k \\
                               & a_0 = a_k = 0, a_j, h_j, e_{jl}, c_{jli} \ge 0
                               & \ \forall & \!\!\!
                               \begin{array}{l}
                                 1 \le j, \, l \le k  \\
                                 1 \le i \le m.
                               \end{array}
\end{array}
\end{equation}

We can derive the upper bound factor-revealing linear program. We would like to
define variables as in equation~\eqref{eq-alg1-323-vardef}. Just using a scale
factor is not sufficient to preserve the variables $a_j$ in
program~\eqref{program-alg1-frlp-relax-ineq-dual}. The variables $a_j$
correspond to the ordering restrictions of primal variables $\alpha_j$ in
program~\eqref{program-alg1-frlp-relax-ineq}, and computational experiments have
indicated that removing such restrictions does not change the optimal value
significantly, for large values of $k$. So, we could just set $a_j = 0$ for
all~$j$. However, we want to preserve such restrictions, as they will shortly be
needed to prove Lemma~\ref{lemma-alg1-gapuplower}. To do this, we can
interpolate the variables of the upper bound factor-revealing program to obtain
the variables of the lower bound program.
Again, we group sets of variables based on their indices.  For that, we denote
the group of a variable of index $n$ as~$\hat{n}$. Let $\hat{n} =
\ceil{\frac{n}{p}}$ and consider prime variables $\gamma', a_j', c_{jli}', e_{jli}',
h_j'$. We obtain a candidate solution for
program~\eqref{program-alg1-frlp-relax-ineq-dual} by defining
\begin{equation}\label{eq-alg1-upfrep-relax-ineq-def}
\gamma = \gamma',
\ \ 
a_j  = p \, a_{\hat{j}}' (p \, \hat{j} - j) (a_{\hat{j}}' - a_{\hat{j}-1}'),
\ \
c_{jli} = \frac{c_{\hat{j}\hat{l}i}'}{p},
\ \
e_{jl} = \frac{e_{\hat{j}\hat{l}}'}{p},
\ \
\mbox{and}
\ \
h_j = \frac{h_{\hat{j}}'}{p}.
\end{equation}

In the following, we will use definition~\eqref{eq-alg1-upfrep-relax-ineq-def}
to obtain a candidate solution for
program~\eqref{program-alg1-frlp-relax-ineq-dual} from a small set of prime
variables. Then, for each constraint of
program~\eqref{program-alg1-frlp-relax-ineq-dual}, we obtain the expression
formed by the dependent terms, and calculate it as a function of the considered
variables. Notice that there is an expression for each primal variable of
program~\eqref{program-alg1-frlp-relax-ineq}. These expressions are analogous to
the primal variables coefficients used in Lemma~\ref{lemma-alg1-323}, thus, for
each primal variable $x$, we say that this is the \emph{coefficient expression}
for~$x$, and we will denote it by $\coef{x}$.

Now we create the minimization upper bound factor-revealing program.  The
objective value is obtained by applying
definition~\eqref{eq-alg1-upfrep-relax-ineq-def} to the objective value of
program~\eqref{program-alg1-frlp-relax-ineq-dual}.  Then, for each group of
coefficient expressions that has the same value, we include a constraint in
the upper bound program that bounds the expression by the independent
term. Notice that each upper bound factor-revealing linear program constraint
may correspond to an arbitrarily large number of constraints of the
factor-revealing linear program. In the following, we calculate and bound each
coefficient expression.

First notice that $a_j - a_{j-1} = a_{\hat{j}}' - a_{\hat{j}-1}'$. To see this,
it is enough to use definition~\eqref{eq-alg1-upfrep-relax-ineq-def}
and consider the cases $\hat{j} = \hat{(j-1)}$,
and $\hat{j} = \hat{(j-1)} + 1$. Now we have:
\begin{align*}
\coef{\alpha_j}
  &=   a_j - a_{j-1}
   +   \sum_{i=1}^{m}     \sum_{l=1}^{k}                                c_{jli}
   -   \sum_{i=1}^{m} B_i \sum_{l=1}^{k}                                c_{lji}
   +                      \sum_{l=j}^{k}                                e_{jl} \\
  &=   a_{\hat{j}}' - a_{\hat{j}-1}'
   +   \sum_{i=1}^{m}     \sum_{l=1}^{pt}                         \frac{c_{\hat{j}\hat{l}i}'}{p}
   -   \sum_{i=1}^{m} B_i \sum_{l=1}^{pt}                         \frac{c_{\hat{l}\hat{j}i}'}{p}
   +                      \sum_{l=j}^{pt}                         \frac{e_{\hat{j}\hat{l}}'}{p} \\
  &\ge a_{\hat{j}}' - a_{\hat{j}-1}'
   +   \sum_{i=1}^{m}     \sum_{l'=1}^{t}                       p \frac{c_{\hat{j}l'i}'}{p}
   -   \sum_{i=1}^{m} B_i \sum_{l'=1}^{t}                       p \frac{c_{l'\hat{j}i}'}{p}
   +                      \sum_{l'=\hat{j}+1}^{t}  p \frac{e_{\hat{j}l'}'}{p} \\
  &= a_{\hat{j}}' - a_{\hat{j}-1}'
   +   \sum_{i=1}^{m}     \sum_{l'=1}^{t}                               c_{\hat{j}l'i}'
   -   \sum_{i=1}^{m} B_i \sum_{l'=1}^{t}                               c_{l'\hat{j}i}'
   +                      \sum_{l'=\hat{j}+1}^{t}                       e_{\hat{j}l'}'
   \ge 1.
\end{align*}

\begin{align*}
\coef{d_j}
  &=   \gamma
      - \sum_{i=1}^{m} C_i \sum_{l=1}^{k} c_{jli}
      - \sum_{i=1}^{m} D_i \sum_{l=1}^{k} c_{lji}
      -                    \sum_{l=1}^{j} e_{jl} \\
  &=  \gamma'
      - \sum_{i=1}^{m} C_i \sum_{l=1}^{pt}                     \frac{c_{\hat{j}\hat{l}i}'}{p}
      - \sum_{i=1}^{m} D_i \sum_{l=1}^{pt}                     \frac{c_{\hat{l}\hat{j}i}'}{p}
      -                    \sum_{l=1}^{j}                      \frac{e_{\hat{j}\hat{l}}'}{p}  \\
  &\ge  \gamma'
      - \sum_{i=1}^{m} C_i \sum_{l'=1}^{t}                    p \frac{c_{\hat{j}l'i}'}{p}
      - \sum_{i=1}^{m} D_i \sum_{l'=1}^{t}                    p \frac{c_{l'\hat{j}i}'}{p}
      -                    \sum_{l'=1}^{\hat{j}} p \frac{e_{\hat{j}l'}'}{p}  \\
  &=   \gamma'
      - \sum_{i=1}^{m} C_i \sum_{l'=1}^{t}                           c_{\hat{j}l'i}'
      - \sum_{i=1}^{m} D_i \sum_{l'=1}^{t}                           c_{l'\hat{j}i}'
      -                    \sum_{l'=1}^{\hat{j}}                     e_{\hat{j}l'}'  \ge 0.
\end{align*}
\begin{align*}
\coef{f}
   = \gamma  - \sum_{j=1}^{k}          h_j
   = \gamma' - \sum_{j=1}^{pt}   \frac{h_{\hat{j}}'}{p}
   = \gamma' - \sum_{j'=1}^{t} p \frac{h_{j'}'}{p}
   = \gamma' - \sum_{j'=1}^{t}         h_{j'} ' \ge 0.
\end{align*}
\begin{align*}
\coef{x_{jl}}
  =       h_j              -       e_{jl}
  = \frac{h_{\hat{j}}'}{p} - \frac{e_{\hat{j}\hat{l}}'}{p}  \ge 0.
\end{align*}

%

We notice that, for each primal variable, the constraint for
its coefficient expression is equivalent to the constraint of any
other primal variable in the same group. For example,
for any pair $\alpha_j$ and $\alpha_l$ such that $\hat{j} = \hat{l}$,
we need to add only one constraint to the upper bound factor-revealing
program; therefore, we need only $t$ constraints for this
kind of primal variable. We remark that
the constraint obtained for $\coef{x_{jl}}$ does not depend on
$p$. Conjoining all different constraints,
and fixing variables $a_1'$ and $a_t'$ to zero,
we obtain program~\eqref{program-alg1-upfrlp-relax-ineq}.
\begin{equation}\label{program-alg1-upfrlp-relax-ineq}\mystyle
\begin{array}{rllll}
\!\!\!\!\! \xx{A1}{t} = \mbox{\rm min}    & \gamma  \\
             \mbox{\rm s.t.}   & a_j  - a_{j-1}
                               + \suml_{i=1}^{m}     \suml_{l=1}^{t}                    c_{jli}
                               - \suml_{i=1}^{m} B_i \suml_{l=1}^{t}                    c_{lji}
                               +                     \suml_{l=j + 1}^{t}   e_{jl}  \ge 1
                               & \forall & 1 \le j \le t \\
                               & \suml_{i=1}^{m} C_i \suml_{l=1}^{t} c_{jli}
                               + \suml_{i=1}^{m} D_i \suml_{l=1}^{t} c_{lji}
                               +                     \suml_{l=1}^{j} e_{lj} \le \gamma
                               & \forall & 1 \le j \le t \\
                               & \suml_{j=1}^{t} h_j \le \gamma
                               &  \\
                               & e_{jl} \le h_j
                               & \forall & 1 \le j \le l \le t \\
                               & a_0 = a_t = 0, \, a_j, \, h_j, \, e_{jl}, \, c_{jli} \ge 0
                               & \forall & \!\!\!
                               \begin{array}{l}
                                 1 \le j, \, l \le t  \\
                                 1 \le i \le m.
                               \end{array}
\end{array}
\end{equation}

Now, we want to use Lemma~\ref{corollary-alg1-inequality} and choose a
set of tuples $(\beta, \gamma, \delta)$ so that the squared metric is
minimally relaxed. To accommodate the premises of
Lemma~\ref{corollary-alg1-inequality}, we solve the dual of the upper
bound factor-revealing LP, so we may use the same cutting plane
strategy used in Lemma~\ref{lemma-alg1-frlp-cut}. The dual is given in
the following.
\begin{equation}\label{program-alg1-upfrlp-relax-ineq-dual}
\begin{array}{rlll}
\xx{A1}{t} = & \mbox{\rm max  }  & \sum_{j=1}^{t} \alpha_j \\
      & \mbox{\rm s.t.}  & f + \sum_{j=1}^t d_j \le 1 & \\
      &                        & \alpha_j \le \alpha_{j+1}
                               & \quad \forall \; 1 \le j < t \\
      &                        & \alpha_j \le B_i  \alpha_l + C_i  d_j + D_i  d_l
                               & \quad \forall \; 1 \le j, \, l \le t ,\quad 1 \le i \le m\\
      &                        & x_{jl} \ge  \alpha_j - d_l
                               & \quad \forall \; 1 \le j < l \le t \\
      &                        & \sum_{l = j}^{t} x_{jl} \le f
                               & \quad \forall \; 1 \le j \le t \\
      &                        & \alpha_j, d_j, f, x_{jl} \ge 0
                               & \quad \forall \; 1 \le j \le l \le t.
\end{array}
\end{equation}

Using the cutting plane strategy with CPLEX we obtain \MAGIC{$\xx{A1}{700}
\approx 2.8697 < 2.87$}.
\end{proof}

If we apply this analysis for the metric case, we obtain an upper bound
factor-revealing program similar to
program~\eqref{program-alg1-upfrlp-relax-ineq-dual}.  The only difference is
that, for the metric case, there are no coefficients $B_l$, $C_l$, and $D_l$.
We use this modified linear program to tighten the approximation factor for
the metric case.

\newcommand{\PasteLemmaFactorUprevealEuc}{
For the \mflp, the approximation factor of $A1$~\cite{JainMMSV03}
is between \MAGIC{$1.814$} and \MAGIC{$1.816$}.
}
\begin{lemma}\label{lemma-factor-upreveal-euc}
\PasteLemmaFactorUprevealEuc
\end{lemma}
\begin{proof}
Let $\zzm{A1}{k}$ be the optimal value of the lower bound factor-revealing program~(5)
in~\cite{JainMMSV03}. The corresponding upper bound factor-revealing program
is:
\begin{equation}\label{program-alg1-upfrlp-metric}
\begin{array}{rlll}
\xxm{A1}{t} = & \mbox{\rm max }  & \sum_{j=1}^{t} \alpha_j \\
      & \mbox{\rm s.t.}  & f + \sum_{j=1}^t d_j \le 1 & \\
      &                        & \alpha_j \le \alpha_{j+1}
                               & \quad \forall \; 1 \le j < t \\
      &                        & \alpha_j \le  \alpha_l + d_j + d_l
                               & \quad \forall \; 1 \le j, \, l \le t\\
      &                        & x_{jl} \ge  \alpha_j - d_l
                               & \quad \forall \; 1 \le j < l \le t \\
      &                        & \sum_{l = j}^{t} x_{jl} \le f
                               & \quad \forall \; 1 \le j \le t \\
      &                        & \alpha_j, d_j, f, x_{jl} \ge 0
                               & \quad \forall \; 1 \le j \le l \le t.
\end{array}
\end{equation}

Numerical computations using CPLEX show that
\MAGIC{$\zzm{A1}{1000} \approx 1.81412 > 1.814 $}, and that
\MAGIC{$\xxm{A1}{1000} \approx 1.81584 < 1.816 $}.
\end{proof}

We notice that the only difference between the upper and lower bound
factor-revealing programs is that the upper bound factor-revealing program
does not contain the restrictions $\alpha_j - d_j \le x_{jj}$ for all $j$. We
exploit the similarity between these programs to bound the gap between their
optimal values. The following lemma is valid for both the metric and squared
metric cases.


\newcommand{\PasteLemmaAlgOneGapuplower}{
Let $\zz{A1}{k}$ be the optimal value of the lower bound factor-revealing
program~\eqref{program-alg1-frlp-relax-ineq} (program~(5) in~\cite{JainMMSV03})
and let $({\bm\alpha}, \mathbf{d}, \mathbf{x}, \mathbf{f})$ be an
optimal solution for program~\eqref{program-alg1-upfrlp-relax-ineq-dual}
(respectively program~\eqref{program-alg1-upfrlp-metric}) with cost value
$\xx{A1}{k}$. If $\varepsilon = \max_{j}\{ \alpha_j - d_j \}$, then $\zz{A1}{k} \ge
\frac{1}{1+\varepsilon} \xx{A1}{k}$.
}
\begin{lemma}\label{lemma-alg1-gapuplower}
\PasteLemmaAlgOneGapuplower%
\end{lemma}%
\begin{proof}
  Let $f'= f+\varepsilon$ and $x'$ be such that $x'_{jl} = x_{jl}$ if $j \ne
  l$, and $x'_{jj} = \max\{0, \alpha_j - d_j\} \geq 0 = x_{jj}$. Observe that 
  $({\bm\alpha}, \mathbf{d}, \mathbf{x'}, \mathbf{f'})$ has objective value 
  $\xx{A1}{k}$ and is a feasible solution for the lower bound 
  factor-revealing program~\eqref{program-alg1-frlp-relax-ineq}, except that 
  it might violate the first restriction of program~\eqref{program-alg1-frlp-relax-ineq} 
  (program~(5) in~\cite{JainMMSV03}, respectively). Indeed, it might be the 
  case that $1 < f'+ \sum_{j=1}^k d_j \leq 1 + \varepsilon$. Now, it is enough 
  to multiply each variable by $\frac{1}{1+\varepsilon}$, and obtain a feasible 
  solution.
\end{proof}

From the last lemma, it is clear that the upper and lower bound
factor-revealing programs yield close values, except for an error factor that
depends only on the variable values of an optimal solution for the upper
bound factor-revealing program. Since the optimal value decreases as the
number of variables $k$ increases, it is reasonable to expect that the value
of both factor-revealing programs become very close as $k$ tends to infinity.
Indeed, for the metric case, it is easy to show that this error vanishes as
$k$ goes to infinity and, therefore, the upper bound and the lower bound
factor-revealing programs converge to the same value, as~$k$ goes to infinity.

\newcommand{\PasteTheoremAlgOneConv}{
Let $\zzm{A1}{k}$ be as in program~(5) in~\cite{JainMMSV03} and let $\xxm{A1}{k}$ be
as in program~\eqref{program-alg1-upfrlp-metric}. Then
$\sup_{k \ge 1} \zzm{A1}{k} = \inf_{k\ge1} \xxm{A1}{k}$.
}
\begin{theorem}\label{theorem-alg1-conv}
\PasteTheoremAlgOneConv
\end{theorem}
\begin{proof}
First notice that, if we double a dual solution of
program~\eqref{program-alg1-upfrlp-metric}, then the obtained solution
for the corresponding minimization upper bound factor-revealing is still
feasible.  Therefore, we may assume that $k$ is arbitrarily large. Consider
an optimal solution of program~\eqref{program-alg1-upfrlp-metric}. We have
that $\alpha_j - d_j \le \alpha_l + d_l$, for every $j$ and $l$. Let $j$ be
such that $\varepsilon = \alpha_j - d_j$ is maximum and add up these
inequalities for all $l$. We get $k \varepsilon = k (\alpha_j - d_j) =
\sum_{l = 1}^{k} (\alpha_j - d_j) \le \sum_{l = 1}^{k} (\alpha_l + d_l) \le
\xxm{A1}{k} + 1 \le 1.816 + 1$. From Lemmas~\ref{lemma-factor-upreveal-euc} and
\ref{lemma-alg1-gapuplower}, we get that $\xxm{A1}{k} \ge \zzm{A1}{k} \ge
\frac{1}{1+\varepsilon} \xxm{A1}{k} \ge \frac{1}{1+2.816/k} \xxm{A1}{k}$. Taking the limit
as $k$ goes to infinity, we get that $\sup_{k \ge 1} \zzm{A1}{k} = \inf_{k\ge1} \xxm{A1}{k}$.
\end{proof}

It would be nice to bound the values of the variables of
program~\eqref{program-alg1-upfrlp-relax-ineq-dual}, as this would suffice
to show that the factor-revealing programs also converge for the squared metric
case. Since the coefficients of the squared triangle inequality involved in
program~\eqref{program-alg1-upfrlp-relax-ineq-dual} are all greater than one,
we cannot use the same approach as in Theorem~\ref{theorem-alg1-conv}.
Although experiments suggest that the value of variable $\alpha_k$ in an
optimal solution decreases as $k$ increases, it does not seem trivial to
determine whether $\alpha_k$ vanishes when $k$ goes to infinity.

\section{Analysis of the second algorithm}\label{sec:segundo}

In this section, we analyze the second algorithm of Jain \etal~\cite{JainMMSV03}
for the squared metric case. The algorithm is essentially the same as
Algorithm~$A1$, but each connected city keeps contributing to unopened
facilities. The contribution of a connected city $j$ to an unopened facility $i$
is the budget that the city would save if facility $i$ were opened. The algorithm,
that is denoted by $A2$, is described in the following.

\paragraph{Algorithm $A2$ $(C,F,c,f)$~\cite{JainMMSV03}}

\begin{enumerate}
\item Set $U:=C$, meaning that every facility starts unopened, and every city unconnected.
  Each city $j$ has some budget $\alpha_j$, initially 0. At every moment, for each
  unopened facility~$i$, if city $j$ is unconnected, then $j$ offers
  $\max(\alpha_j - c_{ij}, 0)$ to $i$, and, if city $j$ is connected to facility $i'$,
  then $j$ offers $\max(c_{i'j} - c_{ij}, 0)$ to $i$.
\item While $U \ne \emptyset$, the budget of each unconnected city is
increased continuously until one of the following events occur:
  \begin{enumerate}
    \item For some unconnected city $j$ and some open facility $i$, $\alpha_j = c_{ij}$.
    In this case, connect city $j$ to facility $i$ and remove $j$ from $U$.
    \item For some unopened facility $i$, the total offer $i$ receives from the
    cities equals the cost~$f_i$ of opening $i$. In this case, open facility $i$, connect
    to $i$ each city $j$ with a positive offer to $i$, and remove each connected
    city from $U$.
  \end{enumerate}
\end{enumerate}

\medskip

For the metric case, the approximation factor is $1.61$. With a completely
analogous reasoning, we obtain the corresponding factor-revealing
program~\eqref{program-alg2-frlp}. The variables are the same as in
program~\eqref{program-alg1-frlp}. The new variable $r_{jl}$ corresponds to
the budget~$\alpha_j$ if city $j$ is connected at the same time as city~$l$,
or corresponds to the distance from~$j$ to the facility to which $j$ is
connected just before $l$ is connected.
\begin{equation}\label{program-alg2-frlp}
\begin{array}{rlll}
\zz{A2}{k} = & \mbox{\rm max }  & \frac{\sum_{j=1}^{k} \alpha_j}{f + \sum_{j=1}^{k} d_j}  \\
      & \mbox{\rm s.t.}  & \alpha_j \le \alpha_{j+1}
                               & \quad \forall \; 1 \le j < k\\
      &                        & r_{jl} \ge r_{j,l+1}
                               & \quad \forall \; 1 \le j < l < k \\
      &                        & \sqrt{\alpha_l} \le \sqrt{r_{jl}} +\sqrt{d_l} + \sqrt{d_j}
                               & \quad \forall \; 1 \le j < l \le k \\
      &                        & \suml_{j = 1}^{l-1} \max(r_{jl} - d_j, 0)
                               + \suml_{j = l}^{k}   \max(\alpha_l - d_j, 0)  \le f
                               & \quad \forall \; 1 \le l \le k \\
      &                        & \alpha_j, \, d_j, \, f, \, r_{j,l} \ge 0
                               & \quad \forall \; 1 \le j \le l \le k.
\end{array}
\end{equation}

We repeat the previous analysis to give lower and upper bounds on the
approximation factor of the second algorithm for the \smflp.

\newcommand{\PasteLemmaAlgTwoFrlpCut}{
$\MAGIC{2.415} \le \sup_{k \ge 1} \zz{A2}{k}  \le \MAGIC{2.425}$.
}
\begin{lemma}\label{lemma-alg2-frlp-cut}
\PasteLemmaAlgTwoFrlpCut
\end{lemma}
\begin{proof}
First, we obtain an upper bound factor-revealing program. See details in
Appendix~\ref{sec-appendix-alg2}. This program is exactly the same as
program~\eqref{program-alg2-frlp}, except that the fourth constraint is replaced with
\[
\sum_{j = 1}^{l-1} \max(r_{jl} - d_j, 0) + \sum_{j = l+1}^{k}   \max(\alpha_l - d_j, 0)  \le f.
\]
Let $\xx{A2}{k}$ be the optimal value of such a program.
With CPLEX we get that \MAGIC{$\zz{A2}{500} \approx 2.41565 > 2.415$},
and that \MAGIC{$\xx{A2}{500} \approx 2.42473 < 2.425$}.
\end{proof}

Solving the upper bound factor-revealing LP obtained for the \mflp for \MAGIC{$k
= 500$}, we may show that the approximation factor of $A2$~\cite{JainMMSV03} is
\MAGIC{$1.602$}. The lower bound factor-revealing program and the maximization
upper bound factor-revealing program are essentially the same, except for
the extra terms of the kind $\max(\alpha_l - d_l)$. Therefore,
Lemma~\ref{lemma-alg1-gapuplower} also holds for such programs. For the metric
case, using a similar analysis to that of Theorem~\ref{theorem-alg1-conv}, one
can show that the lower and the upper bound factor-revealing
programs converge.

\begin{theorem}\label{theorem-alg2-conv}
Let $\zzm{A2}{k}$ be as in program~(25) in~\cite{JainMMSV03} and let
$\xxm{A2}{k}$ be the optimal value of the corresponding upper bound
factor-revealing program obtained by removing the terms of the kind
$\max(\alpha_l - d_l)$ from the fourth restriction. Then $\sup_{k \ge 1}
\zz{A2}{k} = \inf_{k\ge1} \xx{A2}{k}$.
\end{theorem}

\section{Scaling and greedy augmentation}\label{sec:MahdianYZ}

Algorithm $A2$ can be analyzed as a bi-factor approximation algorithm.  The
analysis uses a factor-revealing linear program, and is similar to the previous
analysis.  Mahdian, Ye, and Zhang~\cite{MahdianYZ06} observed that, due to the
asymmetry between the approximation guarantee for the opened facilities cost and
the connections cost, Algorithm $A2$ may be used to open facilities that are
very economical. This gives rise to a two-phase algorithm, denoted here by $A3(\delta)$,
based on scaling the cost of facilities by a constant $\delta \ge 1$, and on the
greedy augmentation technique introduced by Guha and Khuller~\cite{GuhaK98}. The
first phase opens the most economical facilities, and the second phase greedily
includes facilities that reduce the cost of the solution.

\paragraph{Algorithm $A3(\delta)$ $(C,F,c,f)$~\cite{MahdianYZ06}}
\begin{enumerate}
  \item \emph{Scaling}:
    \begin{enumerate}
    \item Scale the facility costs by a factor $\delta$.
    \item Run Algorithm $A2$ on the scaled instance.
    \end{enumerate}
  \item \emph{Greedy augmentation}: \\
    While there are facilities that, if open, reduce the total cost:
    \begin{enumerate}
      \item Compute the gain $g_i$ of opening each unopened facility $i$.
      \item Open a facility $i$ that maximizes the ratio $\frac{g_i}{f_i}$.
    \end{enumerate}
\end{enumerate}

\medskip

In~\cite{MahdianYZ06}, a factor-revealing linear program is used to analyze
Algorithm $A3(\delta)$ with a somewhat different, but equivalent, greedy
augmentation procedure. This was used to balance a bi-factor from Algorithm $A2$
for the \mflp. As noticed by Byrka and Aardal~\cite{ByrkaA2010}, this analysis
is not restricted to Algorithm $A2$, and applies to any bi-factor approximation
for the \flp. Therefore, since it does not depend on the cost function being a
metric, we can use it to balance a bi-factor approximation for the squared
metric case. This result is precisely stated as follows.

\begin{lemma}[\cite{MahdianYZ06}]\label{lemma-mahdian-bifactor}
Consider a $(\gamma_f,\gamma_c)$-approximation for the \flp. Then, for every
$\delta \ge 1$, Algorithm $A3(\delta)$ is a $(\gamma_f + \ln \delta + \varepsilon, 1 +
\frac{\gamma_c-1}{\delta})$-approximation for the \flp.
\end{lemma}

For the metric case, it has been shown that Algorithm $A2$ is a $(1.11,
1.78)$-approxi\-mation. This and Lemma~\ref{lemma-mahdian-bifactor} give a
$1.52$-approximation for the \mflp. For the \smflp, we present an analysis
based on an upper bound factor-revealing program. Using straightforward
calculations, we may obtain the following:

\begin{lemma}\label{lemma-alg3-bifactor}
Let $\gamma_f \ge 1$ be a fixed value and let $\gamma_c = \xx{A2_c}{k}$, where
\begin{equation}\label{program-alg2-upfrlp-bifactor}
\begin{array}{rlll}
\!\!\! \xx{A2_c}{k} = & \mbox{\rm max  }  & \frac{\sum_{j=1}^{k} \alpha_j - \gamma_f f}{\sum_{j=1}^{k} d_j}  \\
      & \mbox{\rm s.t.}  & \alpha_l \le \alpha_{l+1}
                               & \ \forall \; 1 \le l < k\\
      &                        & r_{jl} \ge r_{j,l+1}
                               & \ \forall \; 1 \le j < l < k \\
      &                        & \sqrt{\alpha_l} \le \sqrt{r_{jl}} +\sqrt{d_l} + \sqrt{d_j}
                               & \ \forall \; 1 \le j < l \le k \\
      &                        & \suml_{j = 1}^{l-1} \max(r_{jl} - d_j, 0)
                               + \suml_{j = l+1}^{k}   \max(\alpha_l - d_j, 0)  \le f
                               & \ \forall \; 1 \le l \le k \\
      &                        & \alpha_j, d_j, f, r_{jl} \ge 0
                               & \ \forall \; 1 \le j \le l \le k.
\end{array}
\end{equation}
Then, if $\gamma_c < \infty$, Algorithm $A2$ is a ($\gamma_f, \gamma_c)$-approximation
for the \smflp.
\end{lemma}

The only difference between program~\eqref{program-alg2-upfrlp-bifactor} and the
corresponding lower bound factor-revealing program is the extra term
$\max(\alpha_l - d_l, 0)$ in the lower bound program, which is not in the fourth
constraint of program~\eqref{program-alg2-upfrlp-bifactor}. Again, having a
bound for this term is sufficient to show convergence of the upper and lower
bound factor-revealing programs. For the metric case, this can be done easily.
Notice that we may assume $r_{jl} \le \alpha_j$, so, using a similar analysis to
that of Theorem~\ref{theorem-alg2-conv}, one can show that, if $z_k$ and $x_k$
are solutions for the lower and upper bound programs respectively, then $x_k -
\gamma_f \varepsilon \le z_k \le x_k$, for some $\varepsilon =
\Oh(\frac{1}{k})$.

We observe that program~\eqref{program-alg2-upfrlp-bifactor} is unbounded for
values of $\gamma_f$ close to one. This happens also for the corresponding lower
bound factor-revealing program. This is in contrast to the factor-revealing
programs obtained for the metric case, for which we know that Algorithm~$A2$ is
a $(1,2)$-approximation. In this case, the lower bound program is always
bounded, but the upper bound program is unbounded for $\gamma_f =1$, or for
values close to one. It would be interesting  to strengthen this upper bound
factor-revealing program, so that it could also be used in the analysis also for
$\gamma_f = 1$.

\newcommand{\PasteAlgThreeFactor}{
Algorithm $A3$ is a $2.17$-approximation for the \smflp.
}
\begin{theorem}\label{theorem-alg3-factor}
\PasteAlgThreeFactor
\end{theorem}
\begin{proof}
Consider program~\eqref{program-alg2-upfrlp-bifactor} for \MAGIC{$\gamma_f =
1.45$}. Numerical computations using CPLEX show that \MAGIC{$\xx{A2_c}{300}
\approx 3.40339 < 3.4034$}. From Lemma~\ref{lemma-alg3-bifactor}, we get that
Algorithm~$A2$ is a \MAGIC{$(1.45, 3.4034)$}-approximation for the \smflp. Now,
for \MAGIC{$\delta = 2.0543$}, Lemma~\ref{lemma-mahdian-bifactor} states that
Algorithm~A3 is a \MAGIC{$(2.169\ldots, 2.169\ldots)$}-approximation for the
\smflp.
\end{proof}

In Appendix~\ref{sec-appendix-experiments}, we summarize the results obtained
with CPLEX for the analysis of algorithms~$A1$, $A2$, and $A3$.

\section{An optimal approximation algorithm}\label{sec-optimal}

Byrka and Aardal~\cite{ByrkaA2010} (see also~\cite{ByrkaGS2010}) gave a
$1.5$-approximation for the \mflp, combining a $(1.11, 1.78)$-approximation from
Jain, Mahdian, and Saberi~\cite{Jain2002} and a new analysis of the LP-rounding
algorithm $CS(\gamma)$ of Chudak and Shmoys~\cite{Chudak2004}, that leads to a
$(1.6774, 1.3737)$-approximation. Byrka showed that $CS(\gamma)$ has the optimal
bi-factor approximation $(\gamma, 1 + 2 e^{-\gamma})$ for $\gamma \ge \gamma_0
\approx 1.6774$. By randomly selecting $\gamma$ according to a given probability
distribution, Li~\cite{Li2011} improved this result to $1.488$, that is
currently the best known approximation for the \mflp.

We show that $CS(\gamma)$, when applied to the \smflp, touches its optimal
bi-factor approximation curve $(\gamma, 1 + 8 e^{-\gamma})$ for $\gamma \ge
\gamma_0 \approx 2.00492$. Therefore, we have an $(\alpha, \alpha)$-approximation
for the \smflp, where $\alpha \approx 2.04011$ is the solution of equation
$\gamma = 1 + 8 e^{-\gamma}$. Since $\alpha$ is the approximation lower
bound, this result implies that $CS(\alpha)$, solely used, is an optimal
approximation for the \smflp.

The natural linear program relaxation is given in the following:
\begin{equation}\label{program-rounding}
  \begin{array}{llll}
      \min \         & \sum_{i \in F} y_i f_i + \sum_{j \in C} \sum_{i \in F} x_{ij} c_{ij} \\
      \mbox{  s.t. } & \sum_{i \in F} x_{ij} = 1 && \quad \forall j \in C               \\
                     & x_{ij} \le y_i           && \quad \forall i \in F, j \in C \smallskip  \\
                     & x_{ij}, y_i \ge 0        && \quad \forall i \in F, j \in C.
  \end{array}
\end{equation}

The corresponding integer variables $y_i$ indicate whether facility $i$ is
open, and the corresponding integer variables $x_{ij}$ indicate whether
facility $i$ serves city $j$ in the solution. Algorithm $CS(\gamma)$ may be
summarized as follows. First, a solution $(x^*, y^*)$ of
program~\eqref{program-rounding} is obtained. Then, the fractional opening
variables $y^*_i$ are scaled by a factor $\gamma \ge 1$, $\oy_i = \gamma\,
y^*_i$, and variables $\ox_{ij}$ are defined so that city $j$ is served
entirely by its closest facilities, obtaining a new solution $(\ox,\oy)$.  We
may assume that this solution is \emph{complete}, \emph{i.e.} for every city
$j$ and facility $i$, if $\ox_{ij} > 0$, then $\ox_{ij} = \oy_i$, and that for
every $i$, $\oy_i \le 1$, since, in either case, we can split facility $i$,
and obtain an equivalent instance with these properties. Finally, a clustering
of some of the facilities is obtained according to a given criterion, and a
probabilistic rounding procedure is used to obtain the final solution. For a
detailed description of the algorithm, see~\cite{ByrkaA2010}
(also~\cite{ByrkaGS2010}).

A facility $i$ with $\ox_{ij} > 0$ is called a \emph{close facility} of city~$j$, 
and the set of such facilities is denoted by $C_j$. Similarly, a
facility $i$ with $\ox_{ij} = 0$ but $x_{ij}^* > 0$ is called a \emph{distant
  facility} of~$j$, and the set of such facilities is denoted by $D_j$.
Let $F_j = C_j \cup D_j$.  The analysis of $CS(\gamma)$ uses the notion of
\emph{average distance} between a city $j \in C$ and a subset of facilities
$F' \subseteq F$ such that $\sum_{i \in F'}{\oy_i} > 0$, defined as $ d(j,F') =
\frac{\sum_{i \in F'}{c_{ij} \cdot \oy_i}} {\sum_{i \in F'}{\oy_i}}.  $ For a
city $j$, we also use some definitions from~\cite{ByrkaGS2010}:
the average connection cost, $\djf = d(j, F_j)$; the average distance from
close facilities, $\djc = d(j, C_j)$; the average distance from distant
facilities, $\djd = d(j, D_j)$; the maximum distance from close facilities,
$\djm = \max_{i \in C_j} c_{ij}$; and the irregularity parameter~$\rho_j$,
defined as $\rho_j = {(\djf - \djc)}/{\djf}$ if $\djf > 0$, and $\rho_j = 0$
otherwise.

With these definitions, we can describe the clustering of the facilities. In
each iteration, greedily select a city $j$, called the \emph{cluster center},
such that the sum $\djc{+}\djm$ is minimum, and build a cluster formed by $j$
and its close facilities $C_j$. Remove~$j$ and every other city $j'$ such that
$C_j \cap C_{j'}$ is not empty, and repeat this process until every city is
removed. The set of facilities opened by $CS(\gamma)$ is given by the following
rounding procedure: for each cluster center $j$, open one facility $i$ from
$C_j$ with probability $\ox_{ij} = \oy_i$, and, for each unclustered
facility $i$, open it independently with probability $\oy_i$. Each city is
connected to its closest opened facility.

The following lemma of Byrka and Aardal~\cite{ByrkaA2010} is used to bound the
expected connection cost between a city and the closest facility from a set of
facilities.

\begin{lemma}[\cite{ByrkaA2010}] \label{lemma-exp-distance}
Consider a random vector $y \in \{0,1\}^{|F|}$ produced by Algorithm
$CS(\gamma)$, a subset $A\subseteq F$ of facilities such that $\sum_{i\in A}\bar{y}_i>0$,
and a city $j\in C$. Then, the following holds:
\[
 E\left[\min_{i\in A, y_i=1}c_{ij} \ | \ \sum_{i\in A} y_i \geq 1\right] \leq d(j,A).
\]
\end{lemma}

For a given city $j$, if one facility in $C_j$ or $D_j$ is opened, then
Lemma~\ref{lemma-exp-distance} states that the expected connection cost is bounded
by $\djc$ and $\djd$, respectively. If no facility in $C_j \cup D_j = F_j$ is
opened, then city $j$ can always be connected to one of the close facilities
$C_{j'}$ of the associated cluster center $j'$, with expected connection cost
$d(j,C_{j'} \setminus F_j)$.  Byrka and Aardal~\cite{ByrkaA2010} showed that,
for the \mflp, when $\gamma < 2$, this cost is at most $\djd + \djlm +
\djlc$. Since for the \smflp we need $\gamma > 2$, we will use an improved
version of this lemma by Li~\cite{Li2011}. The adapted lemma for the squared
metric is given in the following. The proof is the same, except that we use the
squared metric property, instead of the triangle inequality.

\newcommand{\PasteLemmaTri}{
Let $j$ be a city and $j'$ be the associated cluster center such that
${C_j \cap C_{j'} \ne \emptyset}$. Then,
$
d(j,C_{j'} \setminus F_j) \le 3 \cdot \left( (2 - \gamma) \djm + (\gamma - 1) \djd + \djlm + \djlc \right).
$
}
\begin{lemma} \label{lemma-tri}
\PasteLemmaTri
\end{lemma}
\begin{proof}
Let $\djjl = \min\limits_{l \in F} (c_{lj} + c_{lj'})$,
that is, the minimum connection cost of a path of length two from $j$ to
$j'$.\footnote{In~\cite{Li2011}, the connection cost $c$ is extended to 
 a distance between $j$ and $j'$, and the triangle inequality is then used to bound this
distance with the connection cost of any path of length two. Here, we make a more
explicit definition to avoid confusion, since the squared metric property is not
sufficient for this purpose.}
Fix a facility $l$ such that $c_{lj} + c_{lj'} = \djjl$.
For each facility $i$ in $C_{j'} \setminus F_j$,
we say that a path $(j,l,j',i)$ is the \emph{center-path} to $i$. The cost
of such center-path to $i$ is defined as $\djjl + c_{ij'}$. Notice that,
using the squared metric property, $c_{ij} \le 3 (\djjl + c_{ij'})$,
and therefore
\begin{align*}
d(j,C_{j'} \setminus F_j)&=    \frac{\sum_{i \in C_{j'} \setminus F_j}{c_{ij} \cdot \oy_i}}
                                    {\sum_{i \in C_{j'} \setminus F_j}{\oy_i}} \\
                         &\le  \frac{\sum_{i \in C_{j'} \setminus F_j}{3 (\djjl + c_{ij'} ) \cdot \oy_i}}
                                    {\sum_{i \in C_{j'} \setminus F_j}{\oy_i}} \\
                         &= 3 \cdot ( \djjl + d(j',C_{j'} \setminus F_j) ).
\end{align*}
That is, $d(j,C_{j'} \setminus F_j)$ is at most three times the average
center-path cost. Following the lines of Li~\cite[Lemma~1]{Li2011},
we know that
$
\djjl + d(j',C_{j'} \setminus F_j)
\le
(2 - \gamma) \djm + (\gamma - 1) \djd + \djlm + \djlc.
$
Therefore, the lemma holds.
\end{proof}

The next lemma follows from Lemma~\ref{lemma-tri}, and is straightforward.

\begin{lemma} \label{lemma-tri-corollary}
$
d(j,C_{j'} \setminus F_j) \le 3 \left( \gamma \djf + (3 - \gamma) \djd \right).
$
\end{lemma}

Now, we can bound the expected facility and connection cost
of a solution generated by~$CS(\gamma)$. The next theorem
is an adapted version of Theorem~2.5 from~\cite{ByrkaGS2010}.

\newcommand{\PasteTheoremFactor}{
For $\gamma \ge 1$, Algorithm $CS(\gamma)$ produces a solution $(x,y)$
for the integer program corresponding to~\eqref{program-rounding} with expected
facility and connection costs
\[
E[y_i f_i] = \gamma \cdot F^*_i,
\mbox{\ \ and \ \ } 
E\left[\min_{i\in F, y_i=1} c_{ij} \right]
   \leq \max\left\{1 + 8e^{-\gamma}, \frac{5 e^{-\gamma}+e^{-1}}{1 - \frac{1}{\gamma}}\right\} \cdot C^*_j,
\]
where $F^*_i = y^*_i f_i$ and $C^*_j = \sum_{i \in F} x^*_{ij} c_{ij}$.
}
\begin{theorem}\label{theorem-factor}
\PasteTheoremFactor
\end{theorem}
\begin{proof}
The expected cost of facility $i$ is
$E[y_i f_i] = \oy_i f_i = \gamma \cdot y^*_i f_i = \gamma \cdot F^*_i$.

If $j$ is a cluster center, one of its close facilities is open, then
the expected connection cost is $\djc \le \djf = C^*_j$.
We may assume that $j$ is not a cluster center. Let
$p_c$ be the probability that the closest facility to $j$ is in $C_j$,
and $p_d$ the probability that it is in $D_j$. If neither case occurs,
then, with probability $p_s = 1 - p_c - p_d$, the closest facility is
in $C_{j'} \setminus F_j$, where $j'$ is the cluster center
associated with $j$.
From the definitions, we have that $\djc = (1- \rho_j) \djf$, \ 
$\djd = (1 + \frac{\rho_j}{\gamma-1})\djf$, and $\rho_j \le 1$. Also, from~\cite{ByrkaA2010}, we know
that $p_s \le e^{-\gamma}$ and $p_c \ge 1 - e^{-1}$.
Combining these facts with Lemmas~\ref{lemma-exp-distance} and~\ref{lemma-tri-corollary},
we obtain
\begin{align*}
E\left[\min_{i\in F, y_i=1} c_{ij}\right]
  &\leq p_c \cdot \djc + p_d \cdot \djd + p_s \cdot 3 \left( \gamma \djf + (3 - \gamma) \djd \right)\\
  &=    \left( (p_c + p_d + 9 p_s) + \frac{(p_c + p_d + 9 p_s)-(p_c + 3 p_s)}{\gamma - 1} \rho_j \right) \djf \\
  &=    \left( (1 + 8 p_s) + \frac{(1 + 8 p_s)-(p_c + 3 p_s)\gamma}{\gamma - 1} \rho_j \right) \djf \\
  &=    \left( (1 + 8 p_s) (1 - \rho_j) + \frac{5 p_s + 1 - p_c}{1 - \frac{1}{\gamma}} \rho_j \right) \djf \\
  &\leq \left( (1 + 8 e^{-\gamma}) (1 - \rho_j) + \frac{5 e^{-\gamma} + e^{-1}}{1 - \frac{1}{\gamma}} \rho_j \right) \djf \\
  &\leq \max\left\{1 + 8e^{-\gamma}, \frac{5 e^{-\gamma}+e^{-1}}{1 - \frac{1}{\gamma}}\right\} \cdot C^*_j.
\end{align*}
\end{proof}

Let $\gamma_0$ be the solution of equation 
$$\left(\frac{5 e^{-\gamma}+e^{-1}}{1 - \frac{1}{\gamma}}\right) = \left(1 +
  8e^{-\gamma}\right).$$ 
For $\gamma \ge \gamma_0 \approx 2.00492$, the maximum connection cost
factor is $1 + 8e^{-\gamma}$, so $CS(\gamma)$ touches the curve
$(\gamma, 1 + 8e^{-\gamma})$, that is, its approximation factor is the
best possible for the \smflp, unless $\PP = \NP$. The next theorem
follows immediately.
\begin{theorem}\label{theorem-opt-alg}
Let $\alpha \approx 2.04011$ be the solution of the equation $\gamma = 1 + 8 e^{-\gamma}$.
Then $CS(\alpha)$ is an $\alpha$-approximation for the \smflp and the approximation factor 
is the best possible unless $\PP = \NP$.
\end{theorem}

\paragraph{Relaxed triangle inequality.}

We notice that the analysis of Lemma~\ref{lemma-alg1-323} and of
Theorem~\ref{theorem-factor} apply to a more general case of the \flp when the
connection cost function satisfies $c_{ij}\leq 3 (c_{ij'} + c_{i'j'} + c_{i'j})$ 
for all facilities~$i$ and~$i'$, and cities~$j$ and $j'$. Charikar
\etal~\cite{Charikar1999-constant} considered a similar relaxed triangle
inequality to extend their constant approximation for the $k$-medians problem with
center costs to the case in which the objective is to minimize the sum of the
squares of the distances of clients to their nearest centers.

We say that a connection cost function $c$ for the \flp satisfies a
$\tau$-relaxed triangle inequality if $c_{ij}\leq \tau \cdot (c_{ij'} + c_{i'j'}
+ c_{i'j}),$ for all~$i, i' \in F$, and $j,j' \in C$. Also, we say that the
subset of the \flp that contains only instances that satisfy the $\tau$-relaxed
triangle inequality is the $\tau$-{\sc relaxed \flp}. The following
theorems extend Theorems~\ref{theorem-inapprox} and~\ref{theorem-factor}.

\begin{theorem}\label{theorem-inapprox-tau}
Let $\gamma_f$ and $\gamma_c$ be positive constants with $\gamma_c < 1 + (3 \tau - 1)
e^{-\gamma_f}$. If there is a $(\gamma_f,\gamma_c)$-approximation for the
$\tau$-relaxed \flp, then $\PP = \NP$. 
\end{theorem}

\begin{theorem}\label{theorem-factor-tau}
For $\gamma \ge 1$, $CS(\gamma)$ is a 
$(\gamma,\max \{1 + (3\tau - 1) e^{-\gamma}, \frac{(2 \tau - 1) e^{-\gamma}+e^{-1}}{1 - \gamma^{-1}}\})$-approxi\-mation
for the $\tau$-relaxed \flp.
\end{theorem}

Let $\alpha(\tau)$ be the solution of equation $\gamma = 1 + (3 \tau -
1)e^{-\gamma}$. It is straightforward to verify
that for $\tau
\ge 2.620\dots$ we have $1 + (3\tau - 1) e^{-\alpha(\tau)} \ge
\frac{(2 \tau - 1) e^{-\alpha(\tau)}+e^{-1}}{1 - \alpha(\tau)^{-1}}$.
Therefore, for $\tau \ge 2.620\dots$, Algorithm $CS(\alpha(\tau))$ has
the best approximation factor for the $\tau$-relaxed \flp.

We say that the {\sc Metric$^\alpha$ \flp}, denoted M$^\alpha$FLP, is the variant of \flp that
considers instances such that the connection cost function is the $\alpha^{\mbox{\small th}}$ power
of a given metric.
We may use the following known fact to derive approximations for M$^\alpha$FLP using 
approximations for $\tau$-relaxed \flp{}s.

\begin{lemma}\label{lemma-power-relax}
Let $A$, $B$, $C$, and $D$ be non-negative numbers
such that $A \leq B + C + D$, and let $\alpha \geq 1$, then 
$A^\alpha \leq 3^{\alpha-1} ( B^\alpha + C^\alpha + D^\alpha )$.
\end{lemma}

This implies that the connection cost function that is the $\alpha^{\mbox{\small th}}$
power of a metric satisfies the $(3^{\alpha-1})$-relaxed triangle inequality, and
therefore M$^\alpha$FLP is a particular case of the $(3^{\alpha-1})$-relaxed
\flp.

\section{The inapproximability threshold for \smflp}\label{sec:complexity}

For the \mflp, Jain \etal~\cite{JainMMSV03} adapted the $1.463$ hardness result
by Guha and Khuller~\cite{GuhaK99}, and showed that no algorithm is a $(\gamma_f,
\gamma_c)$-approximation, with $\gamma_c < 1 + 2 e^{-\gamma_f}$, unless $\NP
\subseteq \DTIME[n^{\Oh(\log \log n)}]$.
Following the lines of Sviridenko (see~Vygen~\cite[Section
  4.4]{Vygen2005}), the condition is strengthened to \emph{unless
  $\PP=\NP$}. We extend these results for the \smflp as follows.

\newcommand{\PasteTheoremInapprox}{
Let $\gamma_f$ and $\gamma_c$ be positive constants with $\gamma_c < 1 + 8
e^{-\gamma_f}$. If there is a $(\gamma_f,\gamma_c)$-approximation for the
\smflp, then $\PP = \NP$. In particular, let $\alpha \approx 2.04011$ be the
solution of the equation $\gamma = 1 + 8 e^{-\gamma}$, then there is no
$\alpha'$-approximation with $\alpha' < \alpha$ for the \smflp unless $\PP =
\NP$.
}
\begin{theorem}\label{theorem-inapprox}
\PasteTheoremInapprox
\end{theorem}
\begin{proof}[Adapted from~\cite{GuhaK99}]
For simplicity, here we show that the lower bound holds unless $\NP \subseteq
\DTIME[n^{\Oh(\log \log n)}]$. If we follow the lines of Sviridenko
(see~Vygen~\cite[Section~4.4]{Vygen2005}), the condition is changed to unless
$\PP=\NP$.

Assume $A$ is a $(\gamma_f,\gamma_c)$-approximation for the \smflp with
$\gamma_c < 1 + 8 e^{-\gamma_f}$. Let $\calj=(\calu, \cals)$ be an instance of
the Set Cover, with $\calu$ being a set of elements, $\cals$ a collection
of subsets of $\calu$ and $n=|\calu|$. We will derive a $(d'\ln
n)$-approximation algorithm for the Set Cover problem, for some $d'<1$.

Let $k$ be the optimal value of $\calj$ for the Set Cover.
If $k$ is not known, one can run this algorithm for $k=1,\ldots,n$ and output
the best solution found.

The algorithm will find a solution for $\calj$ by iteratively solving a
sequence of instances of the \smflp of the form
$\cali^{(j)}=(C^{(j)},F,c,f^{(j)})$, where $F=\cals$ and the initial
set $C^{(1)} = \calu$. For each element $x_j \in S_i$, set $c_{ij} = 1$,
and for each $x_j \not\in S_i$, set $c_{ij}=9$. Note that such $c$ is a squared
metric. Let $n_j=|C^{(j)}|$.  In the $j$th instance, every facility cost is
$f^{(j)} = \gamma \frac{n_j}{k}$, for some positive $\gamma$ to be fixed later.
For each $j$, let $S^{(j)}$ denote the solution for $\cali^{(j)}$ produced
by Algorithm $A$ and let $C^{(j+1)}$ be the elements of $C^{(j)}$ not covered by
any set in $S^{(j)}$.  This process stops when $C^{(j+1)}=\emptyset$ and yields
the solution $S^{(1)}\cup\cdots\cup S^{(j)}$ for $\calj$.

Observe that an optimal solution for $\calj$ is a solution for each
$\cali^{(j)}$ with total facility cost $k\,f^{(j)}$ and connection cost one
for each of the $n_j$ cities. Therefore, $S^{(j)}$ has cost at most
$\gamma_f k f^{(j)} +\gamma_c n_j = (\gamma_f \gamma + \gamma_c) n_j$, because
$f^{(j)} = \gamma \frac{n_j}{k}$.  Let $\beta_j = |S^{(j)}|/k$ and $d_j$ be such
that $d_jn_j$ is the number of elements covered in iteration $j$, that is, the
number of elements of $C^{(j)}$ in the union of the sets in $S^{(j)}$.  Then the
total facility cost of $S^{(j)}$ is $\beta_j k f^{(j)} = \beta_j \gamma\, n_j$. 
Moreover, $d_jn_j$ cities are connected with cost one and the other
$n_j-d_jn_j = (1-d_j)n_j$ cities are connected with cost nine.  Hence the total
cost of $S^{(j)}$ is $\beta_j \gamma n_j + d_j n_j + 9(1 - d_j)n_j = (\beta_j
\gamma + 9 - 8d_j)n_j$. We conclude that $\gamma_f \gamma + \gamma_c \ \ge
\ \beta_j \gamma + 9 - 8d_j.$ So we have that $\gamma_c \ \ge \ (\beta_j
- \gamma_f) \gamma + 9 - 8d_j$.

Let $d<1$ be such that $1+8e^{-\gamma_f/d} > \gamma_c$.
Suppose, for the sake of contradiction, that $d_j \leq 1-e^{-\beta_j/d}$ for
some~$j$. Then
\begin{align*}
\gamma_c  &\ge  (\beta_j - \gamma_f) \gamma  + 9 - 8(1-e^{-\beta_j/d}).
\end{align*}
Considering $\gamma_f$, $\gamma$ and $d$ fixed, the minimum value of the right hand
side is achieved when $\beta_j = d \ln \frac{8}{d \, \gamma}$. Substituting
$\beta_j$ above, we get
\begin{align*}
\gamma_c  &\ge  (d\ln \frac{8}{d\,\gamma} - \gamma_f) \gamma  + 1 + d\,\gamma.
\end{align*}
Considering $d$ and $\gamma_f$ fixed, we choose the value of $\gamma$ that maximizes the
right hand side, that is, $\gamma = \frac{8}{d} e^{-\frac{\gamma_f}{d}}$.
Replacing in the inequality, we obtain $\gamma_c  \ge  1 +
8\,e^{-\frac{\gamma_f}{d}} > \gamma_c$, a contradiction.
So $d_j > 1-e^{-\beta_j/d}$ for every $j$, for this $d<1$.

Following the lines of Guha and Khuller~\cite{GuhaK99},
one can prove that the algorithm described above for the Set Cover is a $(d'\ln n)$-approximation
for some $d'<1$. This  implies that $\NP \subseteq \DTIME[n^{\Oh(\log \log n)}]$.
\end{proof}

\section{Concluding remarks}\label{sec:conclu}

We presented a new technique for deriving upper bound factor-revealing programs, that
can be solved by computer, as an alternative way to obtain an upper bound on the
approximation factors of the corresponding algorithm. This technique allowed us to
tighten the obtained approximation factors, and to simplify the analysis of the three
primal-dual algorithms, when used for both $\smflp$ and $\mflp$ instances. We hope
that this technique can be employed for other problems and algorithms analyzed
through factor-revealing LPs. We also showed that the algorithm of Chudak and
Shmoys\cite{Chudak2004} is a $2.04$-approximation and that it has the best
approximation factor for the \smflp. This is in contrast to the $\mflp$, for which
this algorithm is a $1.575$-approximation and the lower bound by Guha and
Kuller~\cite{GuhaK99} is $1.463$. Also, we note that, although there is an
approximation scheme for Euclidean \flp by Arora~\etal~\cite{AroraRR98}, the analysis
of this algorithm is not valid for non-metric \flp, such as the \seflp. We do not
know whether \seflp has an approximation strictly better than $2.04$.

\bibliographystyle{plain}
\bibliography{squared}

\appendix

\clearpage
\newenvironment{statement}[1]{\noindent{}{\bf #1.} \it}{\normalfont\par}

\section{Square root constraints}
\label{apendice-provas}

The proof of Lemma~\ref{corollary-alg1-inequality} is a straightforward 
consequence of Lemma~\ref{lemma-alg1-metric} and the following result. 

\newcommand{\PasteLemmaIneq}{
Let $A$, $B$, $C$, and $D$ be non-negative numbers.
Then $\sqrt{A} \leq \sqrt{B} + \sqrt{C} + \sqrt{D}$ if and
only if
$A \leq (1 + \beta  + \frac{1}{\gamma}) B
      + (1 + \gamma + \frac{1}{\delta}) C
      + (1 + \delta + \frac{1}{\beta} ) D$
for every positive numbers $\beta$, $\gamma$, and~$\delta$.
In particular, if $\sqrt{A} \leq \sqrt{B}+\sqrt{C}+\sqrt{D}$,
then $A \leq 3B + 3C + 3D$.
}
\begin{lemma}\label{lemma-ineq}
\PasteLemmaIneq
\end{lemma}
\begin{proof}
Suppose $\sqrt{A} \leq \sqrt{B} + \sqrt{C} + \sqrt{D}$.
As $(\sqrt{\beta B} - \sqrt{D/\beta})^2 \geq 0$, we have that
  $2\sqrt{BD} \leq \beta  B + D/\beta $.  Similarly,
  $2\sqrt{CB} \leq \gamma C + B/\gamma$ and
  $2\sqrt{DC} \leq \delta D + C/\delta$. Therefore, if
  $\sqrt{A} \leq \sqrt{B} + \sqrt{C} + \sqrt{D}$, then
\begin{eqnarray*}
A & \leq & (\sqrt{B} + \sqrt{C} + \sqrt{D})^2 \\
  &   =  & B + C + D + 2\sqrt{BD}         + 2\sqrt{CB}          + 2\sqrt{DC} \\
  & \leq & B + C + D + \beta  B + D/\beta + \gamma C + B/\gamma + \delta D + C/\delta \\
  &   =  & (1 + \beta  + \frac{1}{\gamma}) B
         + (1 + \gamma + \frac{1}{\delta}) C
         + (1 + \delta + \frac{1}{\beta} ) D.
\end{eqnarray*}
Choosing $\beta=\gamma=\delta=1$, we obtain $A \leq 3B+3C+3D$.

Now suppose $\sqrt{A} > \sqrt{B} + \sqrt{C} + \sqrt{D}$.
Let $d > 0$ be such that
$A = B + C + D + 2\sqrt{BD} + 2\sqrt{CB} + 2\sqrt{DC} + d$.
Then, $A > (1 + \beta  + \frac{1}{\gamma}) B
      + (1 + \gamma + \frac{1}{\delta}) C
      + (1 + \delta + \frac{1}{\beta} ) D$ is equivalent to
$(\beta  + \frac{1}{\gamma}) B  + (\gamma + \frac{1}{\delta}) C
+ (\delta + \frac{1}{\beta} ) D < 2\sqrt{BD} + 2\sqrt{CB} + 2\sqrt{DC} + d$.
We will analyze the cases in which none, one, two or all numbers $B$, $C$ and $D$
are zero.
Let $\xi$ and $\xi'$ be positive numbers such that $\xi + \xi' < 1$.
\begin{description}
 \item[Case 1:] $B, C, D > 0$.
 Let $\beta = \sqrt{\frac{D}{B}}$,
 $\gamma = \sqrt{\frac{B}{C}}$ and $\delta = \sqrt{\frac{C}{D}}$.
 Then $(\beta  + \frac{1}{\gamma}) B  + (\gamma + \frac{1}{\delta}) C
 + (\delta + \frac{1}{\beta} ) D =  2\sqrt{BD} + 2\sqrt{CB} + 2\sqrt{DC} <
  2\sqrt{BD} + 2\sqrt{CB} + 2\sqrt{DC} + d$.

 \item[Case 2:] $B = 0$ and $C, D > 0$.
 Let $\beta = \frac{D}{\xi d}$,
 $\gamma = \frac{\xi' d}{C}$ and $\delta = \sqrt{\frac{C}{D}}$.
 Then $(\beta  + \frac{1}{\gamma}) B  + (\gamma + \frac{1}{\delta}) C
 + (\delta + \frac{1}{\beta} ) D =  2\sqrt{DC} + (\xi + \xi') d <
  2\sqrt{BD} + 2\sqrt{CB} + 2\sqrt{DC} + d$.

 \item[Case 3:] $B, C = 0$ and $D > 0$.
 Let $\beta = \frac{D}{\xi d}$,
 $\gamma = 1$ and $\delta = \frac{\xi' d}{D}$.
 Then $(\beta  + \frac{1}{\gamma}) B  + (\gamma + \frac{1}{\delta}) C
 + (\delta + \frac{1}{\beta} ) D =  (\xi + \xi') d <
  2\sqrt{BD} + 2\sqrt{CB} + 2\sqrt{DC} + d$.

 \item[Case 4:] $B, C, D = 0$.
 Let $\beta = 1$, $\gamma = 1$ and $\delta = 1$.
 Then $(\beta  + \frac{1}{\gamma}) B  + (\gamma + \frac{1}{\delta}) C
 + (\delta + \frac{1}{\beta} ) D =  0 <
  2\sqrt{BD} + 2\sqrt{CB} + 2\sqrt{DC} + d$.\qedhere
\end{description}
\end{proof}

Observe that the lemma above is constructive in the sense that,
if the given inequality with square roots is not satisfied, 
then shows how to determine a linear inequality that is not 
satisfied. 













\clearpage
\section{Upper Bound Factor-Revealing Program for A2}
\label{sec-appendix-alg2}

Consider tuples $(\beta_i, \gamma_i, \delta_i) \in \RP^3$ and
$B_i = 1 + \beta_i  + \frac{1}{\gamma_i}$,
$C_i = 1 + \gamma_i + \frac{1}{\delta_i}$,
$D_i = 1 + \delta_i + \frac{1}{\beta_i}$ for $1 \le i \le m$.
Using Lemma~\ref{lemma-ineq}, we insert inequalities corresponding to these
tuples, replacing the nonlinear constraint, and obtain $\zz{A2}{k} \le \ww{A2}{k}$, where
$\ww{A2}{k}$ is given by
\begin{equation}\label{program-alg2-frlp-relax-ineq}
\begin{array}{rlll}
\ww{A2}{k} = & \mbox{\rm max }  & \sum_{j=1}^{k} \alpha_j  \\
      & \mbox{\rm s.t.}  & f + \sum_{j=1}^{k} d_j \le 1
                               & \\
      &                        & \alpha_j \le \alpha_{j+1}
                               & \forall \; 1 \le j < k\\
      &                        & r_{jl} \ge r_{j,l+1}
                               & \forall \; 1 \le j < l < k \\
      &                        & \alpha_l \le B_i r_{jl} + C_i d_l + D_i d_j
                               & \forall \; 1 \le j < l \le k , 1 \le i \le m\\
      &                        & r_{jl} - d_j  \le x_{jl}
                               & \forall \; 1 \le j < l \le k \\
      &                        & \alpha_l - d_j  \le x_{jl}
                               & \forall \; 1 \le l \le j \le k \\
      &                        & \sum_{j = 1}^{k}   x_{jl}  \le f
                               & \forall \; 1 \le l \le k \\
      &                        & \alpha_j, d_j, f, r_{jl}  \ge 0
                               & \forall \; 1 \le j \le l \le k \\
      &                        &  x_{jl} \ge 0
                               & \forall \; 1 \le j, l \le k.
\end{array}
\end{equation}

Now, we calculate the dual of program~\eqref{program-alg2-frlp-relax-ineq}
to derive the upper bound factor-revealing linear program.
After that, we calculate its dual program~\eqref{program-alg2-upfrlp-relax-ineq-dual},
in order to use Lemma~\ref{lemma-ineq}, and solve the upper bound factor-revealing
program inserting cutting planes.
We proceed the same way as done in Lemma~\ref{lemma-alg1-upfrlp-cut}.
With similar arguments, we may see that $\zz{A2}{k} \le \zz{A2}{kt}$, for any $t$,
and we assume that $k$ has the form $k = pt$, for some integer $t$.
The dual of linear program~\eqref{program-alg2-frlp-relax-ineq} is given
in the following.

\begin{equation}\label{program-alg2-frlp-relax-ineq-dual}
\begin{array}{rlll}
\ww{A2}{k} = & \mbox{\rm min}         & \gamma  \\
      & \mbox{\rm s.t.}        & a_l  - a_{l-1}
                               + \suml_{i=1}^{m}\suml_{j=1}^{l-1} c_{jli}
                               + \suml_{j=l}^{k}e_{jl} \ge 1
                               & \forall \, 1 \le l \le k \\
      &                        & \gamma
                               - \suml_{i=1}^{m} C_i \suml_{j=1}^{l-1} c_{jli}
                               - \suml_{i=1}^{m} D_i \suml_{j=l+1}^{k} c_{lji}
                               - \suml_{j=1}^{k} e_{lj} \ge  0
                               & \forall \, 1 \le l \le k \\
      &                        & \gamma
                               - \suml_{l=1}^{k} h_l \ge 0
                               &  \\
      &                        & b_{j,l-1} - b_{jl}
                               + e_{jl}
                               - \suml_{i=1}^{m} B_i c_{jli} \ge 0
                               & \forall \, 1 \le j < l \le k \\
      &                        & h_l
                               - e_{jl} \ge 0
                               & \forall \, 1 \le j , l \le k \\
      &                        & a_0 = a_k = b_{ll} = b_{lk} = 0
                               & \forall \, 1 \le l \le k \\
      &                        & a_l, h_l, e_{jl}\ge 0
                               & \forall \, 1 \le l, j \le k \\
      &                        & b_{jl}, c_{jli} \ge 0
                               & \forall \!
                               \begin{array}{l}
                                 1 \le j < l \le k  \\
                                 1 \le i \le m.
                               \end{array}\\
\end{array}
\end{equation}

Now, we may derive the upper bound factor-revealing linear program. Let $\hat{n}
= \ceil{\frac{n}{p}}$ and consider prime variables $\gamma', a_l', b_{jl}',
c_{jli}', e_{jl}', h_l'$. We obtain a candidate solution for
program~\eqref{program-alg2-frlp-relax-ineq-dual} by defining:

\begin{equation}\label{eq-alg2-upfrep-relax-ineq-def}
\begin{array}{c}
\gamma = \gamma',
\quad
a_l = p \, a_{\hat{l}}'  - (p \, \hat{l} - l) (a_{\hat{l}}' - a_{\hat{l}-1}'),
\quad
b_{jl} = b'_{\hat{j},\hat{l}} - \frac{p \, \hat{l} - l}{p} (b'_{\hat{j}\hat{l}} - b'_{\hat{j},\hat{l}-1}),
\\
c_{jll}   = \frac{c_{\hat{j}\hat{l}l}'}{p},
\qquad
e_{jl}     = \frac{e_{\hat{j}\hat{l}}'}{p}
\qquad
\mbox{ and }
\qquad
h_l         = \frac{h_{\hat{l}}'}{p}.
\end{array}
\end{equation}

In the following, we apply definition~\eqref{eq-alg2-upfrep-relax-ineq-def}
and calculate each coefficient expression for program~\eqref{program-alg2-frlp-relax-ineq-dual}.
Again, notice that  $a_l - a_{l-1} = a_{\hat{l}}' - a_{\hat{l}-1}'$,
and that $b_{j,l-1} - b_{jl} = (b'_{\hat{j},\hat{l}-1} - b'_{\hat{j}\hat{l}})/p$.
Also, fix variables $c_{lli}'$ at zero.

\begin{align*}
\coef{\alpha_l}
  &= a_l  - a_{l-1}
   + \sum_{i=1}^{m} \sum_{j=1}^{l-1}       c_{jli}
   +                \sum_{j=l}^{k}         e_{jl} \\
  &= a_{\hat{l}}' - a_{\hat{l}-1}'
   + \sum_{i=1}^{m} \sum_{j=1}^{l-1} \frac{c_{\hat{j}\hat{l}i}'}{p}
   +                \sum_{j=l}^{pt}  \frac{e_{\hat{j}\hat{l}}'}{p} \\
&\ge a_{\hat{l}}' - a_{\hat{l}-1}'
   + \sum_{i=1}^{m} \sum_{j'=1}^{\hat{l}-1}   p \frac{c_{j'\hat{l}i}'}{p}
   +                \sum_{j'=\hat{l} + 1}^{t} p \frac{e_{j'\hat{l}}'}{p} \\
  &= a_{\hat{l}}' - a_{\hat{l}-1}'
   + \sum_{i=1}^{m} \sum_{j'=1}^{\hat{l}-1}                        c_{j'\hat{l}i}'
   +                \sum_{j'=\hat{l} + 1}^{t}                      e_{j'\hat{l}}' \ge 1.
\end{align*}

\begin{align*}
\coef{d_l}
  &=  \gamma
      - \sum_{i=1}^{m} \sum_{j=1}^{l-1} C_i c_{jli}
      - \sum_{i=1}^{m} \sum_{j=l+1}^{k} D_i c_{lji}
      -                \sum_{j=1}^{k} e_{lj} \\
  &=  \gamma'
      - \sum_{i=1}^{m} C_i \sum_{j=1}^{l-1} \frac{c_{\hat{j}\hat{l}i}'}{p}
      - \sum_{i=1}^{m} D_i \sum_{j=l+1}^{k} \frac{c_{\hat{l}\hat{j}i}'}{p}
      -                    \sum_{j=1}^{k}       \frac{e_{\hat{l}\hat{j}}'}{p}  \\
  &\ge \gamma'
      - \sum_{i=1}^{m} C_i \sum_{j'=1}^{\hat{l}}   p \frac{c_{j'\hat{l}i}'}{p}
      - \sum_{i=1}^{m} D_i \sum_{j'=\hat{l}}^{t}   p \frac{c_{\hat{l}j',i}'}{p}
      -                    \sum_{j'=1}^{t}                      p \frac{e_{\hat{l}j'}'}{p}  \\
  &=   \gamma'
      - \sum_{i=1}^{m} C_i \sum_{j'=1}^{\hat{l} - 1}     c_{j'\hat{l}i}'
      - \sum_{i=1}^{m} D_i \sum_{j'=\hat{l}+1}^{t}     c_{\hat{l}j'i}'
      -                \sum_{j'=1}^{t}               e_{\hat{l}j'}' \ge 0.
\end{align*}

\begin{align*}
\coef{f}
   = \gamma  - \sum_{l=1}^{k} h_l
   = \gamma' - \sum_{l=1}^{k} \frac{h_{\hat{l}}'}{p}
   = \gamma' - \sum_{l'=1}^{t} p \cdot  \frac{h_{l'}'}{p}
   = \gamma' - \sum_{l'=1}^{t}                h_{l'} ' \ge 0.
\end{align*}

\begin{align*}
\coef{r_{j,l}}
    =  b_{j,l-1} - b_{jl} + e_{jl}
    - \sum_{i=1}^{r} B_i \, c_{jli}
    = \frac{b'_{\hat{j},\hat{l}-1} - b'_{\hat{j}\hat{l}}}{p} + \frac{e_{\hat{j}\hat{l}}'}{p}
    - \sum_{i=1}^{r} B_i \frac{c_{\hat{j}\hat{l}i}'}{p}  \ge 0.
\end{align*}

\begin{align*}
\coef{x_{jl}}
  =  h_l - e_{j,l}
  =  \frac{h_{\hat{l}}'}{p} -  \frac{e_{\hat{j}\hat{l}}'}{p} \ge 0.
\end{align*}

Conjoining all constraints, the obtained upper bound factor-revealing linear program is:

\begin{equation}\label{program-alg2-upfrlp-relax-ineq}
\begin{array}{rlll}
\xx{A2}{t} = & \mbox{\rm min}         & \gamma  \\
      & \mbox{\rm s.t.}        & a_l  - a_{l-1}
                               + \suml_{i=1}^{m}\suml_{j=1}^{l-1} c_{jli}
                               + \suml_{j=l+1}^{t}e_{jl} \ge 1
                               & \forall \, 1 \le l \le t \\
      &                        & \gamma
                               - \suml_{i=1}^{m} C_i \suml_{j=1}^{l-1} c_{jli}
                               - \suml_{i=1}^{m} D_i \suml_{j=l+1}^{t} c_{lji}
                               - \suml_{j=1}^{t} e_{lj} \ge  0
                               & \forall \, 1 \le l \le t \\
      &                        & \gamma
                               - \suml_{l=1}^{t} h_l \ge 0
                               &  \\
      &                        & b_{j,l-1} - b_{jl}
                               + e_{jl}
                               - \suml_{i=1}^{m} B_i c_{jli} \ge 0
                               & \forall \, 1 \le j < l \le t \\
      &                        & h_l
                               - e_{jl} \ge 0
                               & \forall \, 1 \le j , l \le t \\
      &                        & a_0 = a_t = b_{ll} = b_{lt} = 0
                               & \forall \, 1 \le l \le t \\
      &                        & a_l, h_l, e_{jl}\ge 0
                               & \forall \, 1 \le l, j \le t \\
      &                        & b_{jl}, c_{jli} \ge 0
                               & \forall \!
                               \begin{array}{l}
                                 1 \le j < l \le k  \\
                                 1 \le i \le m.
                               \end{array}\\
\end{array}
\end{equation}

Finally, calculating the dual of program~\eqref{program-alg2-upfrlp-relax-ineq},
we obtain program~\eqref{program-alg2-upfrlp-relax-ineq-dual}.

\begin{equation}\label{program-alg2-upfrlp-relax-ineq-dual}
\begin{array}{rlll}
\xx{A2}{t} = & \mbox{\rm max } & \sum_{j=1}^{t} \alpha_j  \\
      & \mbox{\rm s.t.}        & f + \sum_{j=1}^{t} d_j \le 1
                               & \\
      &                        & \alpha_j \le \alpha_{j+1}
                               & \forall \; 1 \le j < t\\
      &                        & r_{jl} \ge r_{j,l+1}
                               & \forall \; 1 \le j < l < t \\
      &                        & \alpha_l \le B_i r_{jl} + C_i d_l + D_i d_j
                               & \forall \; 1 \le j < l \le t , 1 \le i \le m\\
      &                        & r_{jl} - d_j  \le x_{jl}
                               & \forall \; 1 \le j < l \le t \\
      &                        & \alpha_l - d_j  \le x_{jl}
                               & \forall \; 1 \le l < j \le t \\
      &                        & \sum_{j = 1}^{t}   x_{jl}  \le f
                               & \forall \; 1 \le l \le t \\
      &                        & \alpha_j, d_j, f, r_{jl}  \ge 0
                               & \forall \; 1 \le j \le l \le t \\
      &                        &  x_{jl} \ge 0
                               & \forall \; 1 \le j, l \le t.
\end{array}
\end{equation}

\clearpage
\section{Experimental results}
\label{sec-appendix-experiments}

In Table~\ref{table-ex1}, we present computational results using CPLEX for the
lower bound (column $\zz{A1}{k}$) and upper bound (column $\xx{A1}{k}$) for the
approximation factor of Algorithm $A1$.
In Table~\ref{table-ex2}, we present lower and upper bounds on the approximation
factor of Algorithm $A2$ (columns $\zz{A2}{k}$ and $\xx{A2}{k}$, respectively).
In Table~\ref{table-ex3}, we present computational results for
program~\eqref{program-alg2-frlp} when $\gamma_f = 1.45$, and the approximation
factor obtained from Lemma~\ref{lemma-mahdian-bifactor}. The chosen $\delta$ is
given by the solution of equation $\gamma_f + \ln \delta  = 1 +
\frac{\gamma_c-1}{\delta}$, that is, $\delta =
e^{W_0((\gamma_c-1)e^{\gamma_f-1})-(\gamma_f-1)}$.
Figure~\ref{fig-factor} shows the trade-off between connection and facility
costs approximation guarantees for the Algorithm $A2$, and Figure~\ref{fig-trend}
shows the trend of obtained factor for Algorithm $A3$ as we vary the value of
$\gamma_f$, when $k = 50$.

\noindent{}
\hfil
\begin{minipage}[t][][t]{0.4\textwidth}
  \captionof{table}{Solutions of the factor-revealing programs for $A1$.\label{table-ex1}}
  \hfil
  \begin{tabular}{| r | c|c |}
    \hline  \multicolumn{1}{|c|}{$k$}      &  $\zz{A1}{k}$   & $\xx{A1}{k}$    \\ \hline
    10   & 2.57261 &  3.18162 \\ \hline
    20   & 2.71704 &  3.01717 \\ \hline
    50   & 2.80540 &  2.92579 \\ \hline
    100  & 2.83534 &  2.89553 \\ \hline
    200  & 2.85034 &  2.88046 \\ \hline
    300  & 2.85532 &  2.87543 \\ \hline
    400  & 2.85782 &  2.87292 \\ \hline
    500  & 2.85930 &  2.87142 \\ \hline
    600  & 2.86029  &  2.87041 \\ \hline
    700  & 2.86099 &  2.86970 \\ \hline
  \end{tabular}
  \hfil
\end{minipage}
\hfil
\begin{minipage}[t][][t]{0.4\textwidth}
  \captionof{table}{Solutions of the factor-revealing programs for $A2$.\label{table-ex2}}
  \hfil
  \begin{tabular}{|r|c|c|}
    \hline  \multicolumn{1}{|c|}{$k$}      &  $\zz{A2}{k}$   & $\xx{A2}{k}$    \\ \hline
    10   & 2.20702 &  2.65131 \\ \hline
    20   & 2.30987 &  2.53301 \\ \hline
    50   & 2.37551 &  2.46544 \\ \hline
    100  & 2.39773 &  2.44278 \\ \hline
    200  & 2.40894 &  2.43150 \\ \hline
    300  & 2.41267 &  2.42775 \\ \hline
    400  & 2.41453 &  2.42586 \\ \hline
    500  & 2.41565 &  2.42473 \\ \hline
  \end{tabular}
  \hfil
\end{minipage}
\hfil

\bigskip

\noindent{}
\hfil
\begin{minipage}[t][][t]{0.55\textwidth}
  \captionof{table}{Solutions of connection factor-revealing programs for $A2$, and obtained factor for $A3$.}
  \label{table-ex3}
  \hfil
    \begin{tabular}{| r | c | c | c |}
      \hline  \multicolumn{1}{|c|}{$k$} &  $\xx{A2_c}{k}$   & best $\delta$  & factor  \\ \hline
      10   & 4.02931 & 2.33433 & 2.29772 \\ \hline
      20   & 3.64790 & 2.16561 & 2.22270 \\ \hline
      50   & 3.48465 & 2.09159 & 2.18792 \\ \hline
      100  & 3.43524 & 2.06895 & 2.17704 \\ \hline
      200  & 3.41127 & 2.05793 & 2.17170 \\ \hline
      300  & 3.40339 & 2.05430 & 2.16993 \\ \hline
    \end{tabular}
  \hfil
\end{minipage}
\hfil

\bigskip

\noindent{}
\hfil
\begin{minipage}[t]{0.4\textwidth}
  \hfil
  \includegraphics[scale=0.48]{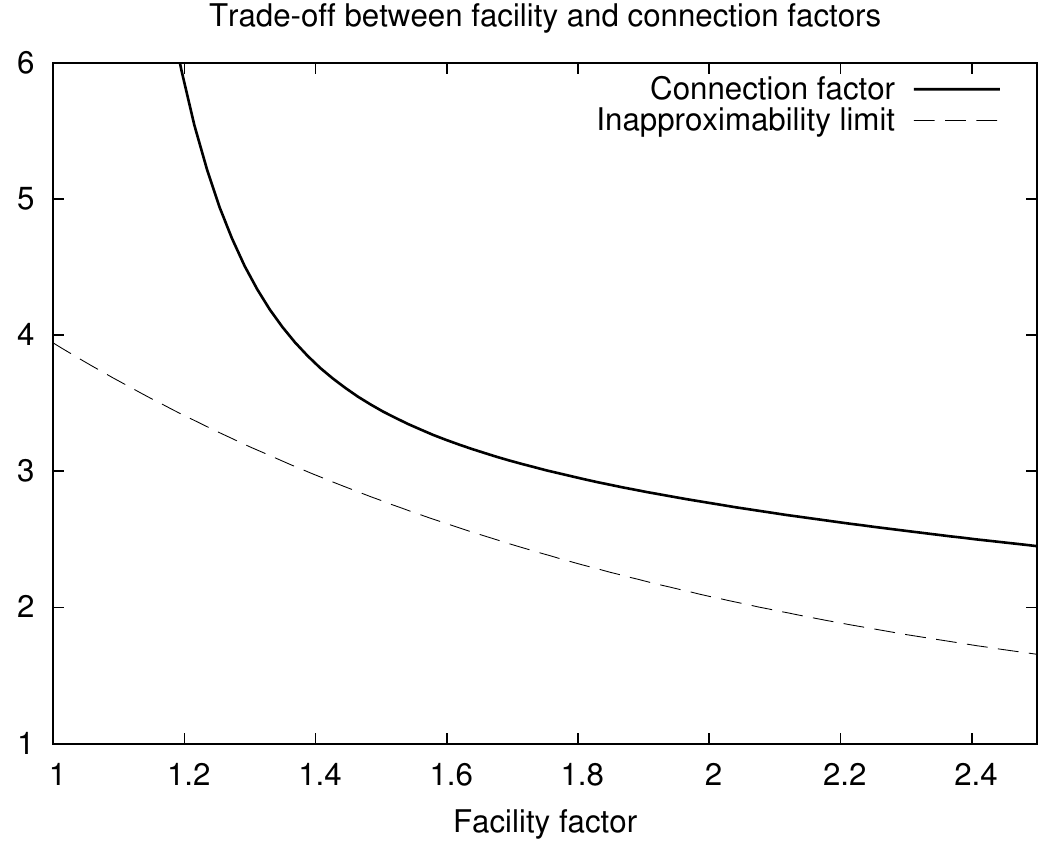}
  \hfil
  \captionof{figure}{Trade-off between connection and facility approximation factors.\label{fig-factor}}
\end{minipage}
\hfil
\begin{minipage}[t]{0.4\textwidth}
  \hfil
  \includegraphics[scale=0.48]{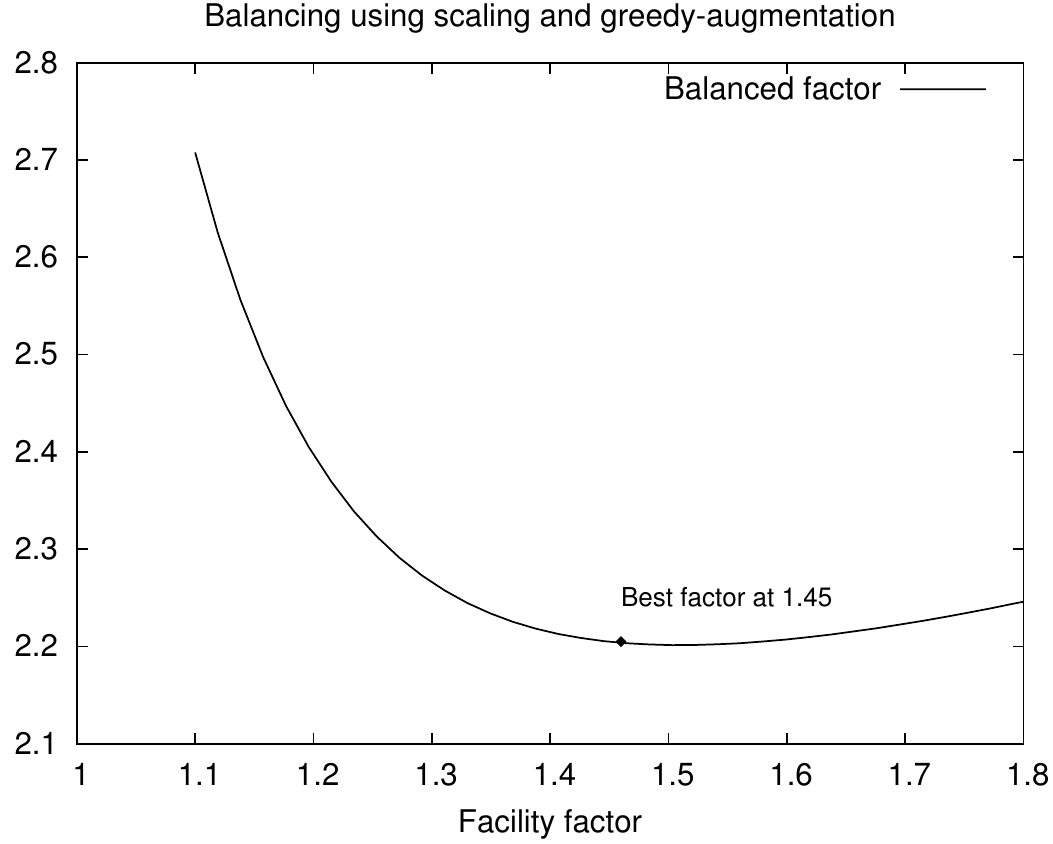}
  \hfil
  \captionof{figure}{Trend of the obtained balanced approximation factors.\label{fig-trend}}
\end{minipage}
\hfil

\end{document}